\documentclass[11pt]{article}

\title{\textbf{Improved Deterministic $(\Delta+1)$-Coloring \\ in Low-Space \MPC}\thanks{This work is partially supported by a Weizmann-UK Making Connections Grant,  the Centre for Discrete Mathematics and its Applications (DIMAP), IBM Faculty Award, EPSRC award EP/V01305X/1, European Research Council (ERC) Grant No. 949083, and the European Union's Horizon 2020 programme under the Marie Sk{\l}odowska-Curie grant agreement No 754411.}}
	
\author{\textbf{Artur Czumaj}
\thanks{Department of Computer Science and Centre for Discrete Mathematics and its Applications (DIMAP), University of Warwick. Email: A.Czumaj@warwick.ac.uk.
}
	\and
\textbf{Peter Davies}
\thanks{Institute of Science and Technology Austria (IST Austria). Email: peter.davies@ist.ac.at.
}
	\and
\textbf{Merav Parter}
\thanks{Weizmann Institute of Science, Rehovot, Israel. Email: merav.parter@weizmann.ac.il.
}
}
	
\author{\textbf{Artur Czumaj} \\
    University of Warwick \\
    A.Czumaj@warwick.ac.uk
	\and
\textbf{Peter Davies} \\
    IST Austria \\
    peter.davies@ist.ac.at
	\and
\textbf{Merav Parter} \\
    Weizmann Institute of Science \\
    merav.parter@weizmann.ac.il
}

\usepackage{amsmath,amsthm}
\usepackage{amssymb}
\usepackage{algorithm,algpseudocode}
\usepackage{bbm,bm}
\usepackage{fullpage,color}
\usepackage{color}
\usepackage{xspace}

\newtheorem{theorem}{Theorem}
\newtheorem{corollary}[theorem]{Corollary}
\newtheorem{lemma}[theorem]{Lemma}
\newtheorem{proposition}[theorem]{Proposition}

\newtheorem{definition}[theorem]{Definition}
\newtheorem{observation}[theorem]{Observation}

\newtheorem{claim}[theorem]{Claim}

\newtheorem{framework}[theorem]{Framework}

\usepackage{framed}
\usepackage[framemethod=tikz]{mdframed}

\usepackage{hyperref}

\newcommand{\COMMENTED}[1]{{}}
\newcommand{\junk}[1]{\COMMENTED{#1}}
\newcommand{\hide}[1]{{}}

\newcommand{\Prob}[1]{\bm{Pr}\left[#1\right]}	
\newcommand{\Exp}[1]{\bm{E}\left[#1\right]}
\newcommand{\poly}{\operatorname{\text{{\rm poly}}}}
\newcommand{\cj}{\ensuremath{c}\xspace}
\newcommand{\cA}{\mathcal{A}}

\newcommand{\OneShotColoring}{\mathsf{OneShotColoring}}
\newcommand{\ColorBidding}{\mathsf{ColorBidding}}
\newcommand{\DenseColoringStep}{\mathsf{DenseColoringStep}}
\newcommand{\Pal}{\Phi}
\newcommand{\nat}{\ensuremath{\mathbb{N}}}
\newcommand{\eps}{\ensuremath{\epsilon}}
\newcommand{\local}{\textsf{LOCAL}\xspace}
\newcommand{\LOCAL}{\local}
\newcommand{\congest}{\textsf{CONGEST}\xspace}
\newcommand{\congc}{\textsf{CONGESTED CLIQUE}\xspace}
\newcommand{\MPC}{\textsf{MPC}\xspace}
\newcommand{\sparam}{\ensuremath{\phi}\xspace}
\newcommand{\dist}{\mbox{\rm dist}}

\newcount\shortyear\newcount\shorthour\newcount\shortminute
\shorthour=\time\divide\shorthour by 60\shortyear=\shorthour
\multiply\shortyear by 60\shortminute=\time\advance\shortminute by
-\shortyear \shortyear=\year\advance\shortyear by -1900

\def\zeit{\number\shorthour:\ifnum\shortminute<10 0\number\shortminute
	\else\number\shortminute\fi}

\begin{document}

\maketitle

\begin{abstract}
We present a \emph{deterministic} $O(\log \log \log n)$-round \emph{low-space} Massively Parallel Computation (MPC) algorithm for the classical problem of $(\Delta+1)$-coloring on $n$-vertex graphs. In this model, every machine has a sublinear local memory of size $n^{\sparam}$ for any arbitrary constant $\sparam \in (0,1)$. Our algorithm works under the relaxed setting where each machine is allowed to perform exponential (in $n^{\sparam}$) local computation, while respecting the $n^{\sparam}$ space and bandwidth limitations.

Our key technical contribution is a novel derandomization of the ingenious $(\Delta+1)$-coloring \textsf{LOCAL} algorithm by Chang-Li-Pettie (STOC 2018, SIAM J. Comput. 2020). The Chang-Li-Pettie algorithm runs in $T_{local}=poly(\log\log n)$ rounds, which sets the state-of-the-art randomized round complexity for the problem in the local model. Our derandomization employs a combination of tools, most notably pseudorandom generators (PRG) and bounded-independence hash functions.

The achieved round complexity of $O(\log\log\log n)$ rounds matches the bound of $\log(T_{local})$, which currently serves an upper bound barrier for all known \emph{randomized} algorithms for locally-checkable problems in this model. Furthermore, no deterministic sublogarithmic low-space MPC algorithms for the $(\Delta+1)$-coloring problem were previously known.
\junk{
We present an $O(\log \log \log n)$-round \emph{deterministic} low-space Massively Parallel Computation algorithm for $(\Delta+1)$-coloring on $n$-vertex graphs. In this model, every machine has a sub-linear memory of $n^{\delta}$ for any positive $\delta< 1$. Our algorithm works under the relaxed setting where each machine is allowed to perform an unbounded local computation, while respecting its $n^{\delta}$ space and bandwidth limitations.

Our key technical contribution is a novel derandomization of the ingenious $(\Delta+1)$ coloring \emph{local} algorithm by Chang-Li-Pettie (STOC 2018,  SIAM J. Comput. 2020). The Chang-Li-Pettie algorithm runs in $T_{local}=\poly(\log\log n)$ rounds, which sets the state-of-the-art randomized round complexity for the problem in the local model. Our derandomization employs a combination of tools, most notably pseudorandom generators (PRG) and bounded-independence hash functions.

The achieved round complexity of $O(\log \log \log n)$ rounds matches the bound of $\log(T_{local})$,
which currently serves an upper bound barrier for all known \emph{randomized} algorithms in this model, as observed and conditionally proven by Ghaffari, Kuhn, and Uitto (FOCS 2019). No sub-logarithmic-round low-space MPC algorithms for the problem have been known before.
}
\end{abstract}

\newpage 

%
%
%

\section{Introduction}

In this paper, we study the deterministic complexity of the $(\Delta+1)$ (list) coloring problem in the low-space \MPC setting. The \emph{Massively Parallel Computation (\MPC)} model, introduced by Karloff, Suri and Vassilvitskii \cite{KSV10}, is a nowadays standard theoretical model for parallel algorithms. This model shares many similarities to earlier models of parallel computation (e.g., PRAM), and it is also closely related to various distributed models, such as the \local and the \congc\ models. We focus on the \emph{low-space} MPC regime in which machines have
space $n^{\sparam}$ for any arbitrary constant $\sparam\in (0,1)$, where $n$ is the number of nodes in the graph. This model has attracted a lot of attention recently \cite{BKM20,BBDFHKU19,CFGUZ19,CDP20a,CDP20b,GGKMR18,GGJ20,GhaffariJN20,GKU19,GU19,GSZ11,KothapalliPP20}, especially in the context of \emph{local} graph problems. Recent works have provided many \emph{randomized} algorithms with sub-logarithmic round complexities, for fundamental local graph problems such as maximal matching, maximal independent set and $(\Delta+1)$ coloring. However, much less is known on the corresponding deterministic complexity of these problems. In particular, to this date no \emph{sublogarithmic deterministic} algorithm is known, in the low-space \MPC model, for any of the canonical symmetry breaking problems.

We study deterministic low-space \MPC algorithms for the $(\Delta+1)$ (list) coloring problem, which is arguably among the most fundamental graph problems in parallel and distributed computing with numerous implications. In this problem, we are given an input graph $G=(V,E)$ with maximum degree $\Delta$, for which every vertex has a palette (list) of $\Delta+1$ colors. The goal is to compute a legal vertex coloring, namely, where no two neighbors have the same color, in which each node is assigned a color from its own palette. A sequence of recent exciting breakthrough results have led to a dramatic improvement in the randomized and the deterministic complexity of the problem, in the classical distributed models, as we highlight next.


\paragraph{$(\Delta+1)$ Coloring in the \local Model:}
The \local model has been introduced by Linial \cite{Linial92} with the purpose of developing symmetry breaking methodologies in decentralized networks. In this model,
each node in the communication graph is occupied by a processor. The processors communicate in synchronous message passing rounds where per round each processor can send one message to each of its neighbors in the network.
Since its introduction, the model has focused on four canonical problems and their variants: maximal independent sets, $(\Delta+1)$ coloring, and their \emph{edge} analogs, namely, maximal matching and edge coloring. As this model abstracts away congestion issues, it provides the most convenient platform for studying the locality aspects of symmetry breaking.

The study of the $(\Delta+1)$ coloring problem in this model has quite a long history with several important milestones. We first focus on randomized algorithms, and then address the deterministic aspects of the problem. Logarithmic solutions for $(\Delta+1)$ coloring are known since the 80's, e.g., by the classical Luby-MIS algorithm \cite{Luby86}. Barenboim et al. \cite{BEPS16} presented the shattering technique, which in the context of coloring, reduces the problem, within $O(\log \Delta)$ randomized rounds, into independent subproblems of $\poly\log n$ size, which can be then solved deterministically.  Harris, Schneider and Su \cite{harris2016distributed} presented a new graph decomposition technique that provided the first sublogarithmic solution for the problem. Finally, in a subsequent remarkable breakthrough result,  Chang, Li and Pettie (CLP) \cite{ChangLP18,chang2020distributed} presented an $O(Det_d(\poly \log n))$-round solution for the problem, where $Det_d(n')$ is the deterministic complexity of the $(\deg+1)$-list coloring problem on an $n'$-vertex graph. In the latter problem, every vertex $v$ has a palette of only $\deg(v)+1$ colors. Their upper bound \emph{morally} matches the lower bound of $\Omega(Det(\poly \log n))$ rounds shown by Chang, Kopelowitz and Pettie \cite{CKP19} with the only distinction being that $Det(n')$ is the deterministic complexity of the $(\Delta+1)$-list coloring problem. Combining the CLP algorithm with the recent deterministic network decomposition result of Rozho{\v{n}} and Ghaffari \cite{RG20}, yields an $\poly(\log\log n)$-round algorithm for the $(\Delta+1)$ list coloring problem, which sets the state-of-the-art bound for the problem. The randomized complexity for the related $(\deg+1)$ coloring is $O(\log \Delta)+\poly(\log\log n)$ by \cite{BEPS16,RG20}.


Obtaining \emph{deterministic} coloring solutions of polylogarithmic-time has been one of the most major open problems in the area. This was resolved recently by the groundbreaking network decomposition result of Rozho{\v{n}} and Ghaffari \cite{RG20}. Even more recently, Ghaffari and Kuhn \cite{GhaffariKuhnDetCol21} improved the time bounds into $O(\log^2\Delta\log n)$ rounds, by using the more direct approach of rounding fractional color assignments. Due to the shattering-based structure of the CLP solution, any deterministic algorithm for the problem immediately improves also the (randomized) CLP bound.

\paragraph{$(\Delta+1)$ Coloring in the \congc\ Model:}
In the \congc\ model, introduced by Lotker, Pavlov and Patt-Shamir \cite{LotkerPPP03,LotkerPPP05}, the network is represented as a fully connected graph, where each node is occupied by a machine which stores the node' edges. The machines communicate in an all-to-all fashion, where in each round, every pair of machines can exchange $O(\log n)$ bits of information. The local memory and computation power are assumed to be \emph{unlimited}. As we will see, this model is considerably more relaxed than the low-space \MPC model that we consider in this paper.

There has been a sequence of recent results concerning the randomized complexity of the $(\Delta+1)$ coloring in this model. Parter \cite{Parter18} presented an $O(\log\log\Delta)$-round algorithm for the problem that is based on combining the CLP algorithm with a recursive degree reduction. By employing a palette sparsification technique, Parter and Su \cite{ParterSu18} improved the complexity into  $O(\log^* \Delta)$ rounds. Finally, the randomized complexity of the problem has been settled into $O(1)$ rounds, by Chang et al. \cite{CFGUZ19}. Their algorithm also supports the list variant of the problem, by employing a new randomized partitioning of both the nodes and their \emph{colors}.
Recently, Czumaj, Davies and Parter \cite{CDP20b} provided a simplified $O(1)$-round deterministic algorithm for the problem. In contrast to prior works, their algorithm is not based on the CLP algorithm. Prior deterministic (logarithmic) bounds were also given by Parter \cite{ParterSu18} and Bamberger, Kuhn and Maus \cite{BKM20}. 

\paragraph{$(\Delta+1)$ Coloring in the Low-Space \MPC Model:}
In the \MPC model, there are $M$ machines and each of them has $S$ words of space. Initially, each machine receives its share of the input. In our case, the input is a collection $V$ of nodes and $E$ of edges and each machine receives approximately $\frac{n+m}{M}$ of them (divided arbitrarily), where $|V| = n$ and $|E| = m$. The computation proceeds in synchronous \emph{rounds} in which each machine processes its local data and performs an arbitrary local computation on its data without communicating with other machines. At the end of each round, machines exchange messages. Each message is sent only to a single machine specified by the machine that is sending the message. All messages sent and received by each machine in each round have to fit into the machine's local space. Hence, their total length is bounded by $S$. This, in particular, implies that the total communication of the \MPC model is bounded by $M \cdot S$ in each round. The messages are processed by recipients in the next round.  At the end of the computation, machines collectively output the solution. The data output by each machine has to fit in its local space of $S$ words.

We focus on the \emph{low-space} (also called \emph{strongly sublinear}) regime where $S=n^{\sparam}$ for any given constant $\sparam \in (0,1)$. A major challenge underlying this setting is that the local space of each machine might be too small to store all the edges incident to a single node. This poses a considerable obstacle for simulating \local\ algorithms compared to the linear space regime. To overcome this barrier, both randomized and deterministic algorithms in this model are based on graph sparsification techniques.

Chang et al. \cite{CFGUZ19} presented the first randomized algorithm for $(\Delta+1)$ coloring in this model, which as described before\footnote{Their \MPC algorithm is similar to their \congc\ algorithm.}, employs a random node and palette partitioning which breaks the problem into independent coloring instances. A sparsified variant of the CLP algorithm is then applied on each of the instances, in parallel. This approach when combined with the network decomposition result of \cite{RG20} provides an $O(\log\log\log n)$ round algorithm, which is currently the state-of-the-art bound for the problem.

The deterministic complexity of the $(\Delta+1)$ coloring in low-space \MPC has been studied independently by Bamberger, Kuhn and Maus \cite{BKM20} and by Czumaj, Davies and Parter \cite{CDP20b}: \cite{BKM20} presented an $O(\log^2 \Delta+\log n)$ round solution for the $(\deg+1)$ list coloring problem; \cite{CDP20b} presented an $O(\log \Delta+\log\log n)$-round algorithm for the $(\Delta+1)$ list coloring problem. No sublogarithmic bounds are currently known.  To the best of our knowledge, the only sublogarithmic deterministic solutions in this model are given for the ruling set problem\footnote{In the $\beta$ ruling set problem, it is required to compute an independent set $S$ such that every vertex as a $\beta$-hop neighbor in $S$.} by Kothapalli, Pai and Pemmaraju~\cite{KothapalliPP20}.

\paragraph{On the connection between low-space \MPC and \local models:} Many of the existing algorithms for local problems in the low-space \MPC model are based on \local algorithms for the corresponding problems, e.g., \cite{BBDFHKU19,CFGUZ19,GGKMR18,Onak18}. Specifically, using the graph exponentiation technique, $T$-round \local algorithms can be simulated within $O(\log T)$ \MPC rounds, provided that the $T$-balls of each node fits the space of the machine. Since in many cases the balls are too large, this technique is combined with other round compression approaches, such as graph sparsification, that are aimed at simulating many \local rounds using few \MPC rounds. The upper bound limit of all current approaches is $O(\log T_{\local})$ \MPC rounds, where $T_{local}$ is the \local complexity of the problem.

In a recent inspiring paper, Ghaffari, Kuhn and Uitto \cite{GKU19} established a connection between these two models in the reverse direction (see also a revised framework in \cite{CDP21a}). They presented a general technique that allows one to lift \emph{lower bound} results in the \local model into lower bounds in the low-space \MPC model, conditioned on the connectivity conjecture. Using this approach they provided conditional lower bounds of $\Omega(\log(T_{\local}))$ \MPC rounds given an $\Omega(T_{\local})$-round \local lower bound for the corresponding problem. While the original framework from \cite{GKU19} holds only for randomized algorithms, the revised framework in \cite{CDP21a} applies also to deterministic algorithms. One caveat of these results is that they hold only for the class of \emph{component-stable} \MPC algorithms. Roughly speaking, in this class of algorithms the output of a node depends only on its connected component. We note that the deterministic algorithms presented in this paper are \emph{not} component stable. Our algorithm matches the logarithm of the randomized \local complexity of the $(\Delta+1)$ list coloring problem, which is currently an upper bound limit even for randomized algorithms, for most of the canonical local graph problems.

\subsection{Our Results}

Our key result is an $O(\log\log\log n)$-time deterministic algorithm for the $(\Delta+1)$ (list) coloring problem in the low-space MPC model. Our algorithm employs exponential (in the local space, i.e. $exp(n^{\sparam})$) local computation at each machine, while respecting its (strongly sublinear) space requirement.

\begin{mdframed}[hidealllines=true,backgroundcolor=gray!25]
\vspace{-8pt}
\begin{theorem}\label{thm:det-coloring-main}
There exists a deterministic algorithm that, for every $n$-vertex graph $G=(V,E)$ with maximum degree $\Delta$, computes a
$(\Delta+1)$ (list) coloring for $G$ using $O(\log\log\log n)$ rounds, in the low-space \MPC\ model with global space $\widetilde{O}(|E|+n^{1+\sparam})$. The algorithm employs exponential time local computation (in $n^{\sparam}$) at each machine while respecting the space and bandwidth limitations.
\end{theorem}
\end{mdframed}

Alternatively, we can also state the result as a non-explicit, non-uniform, polynomial-computation deterministic low-space \MPC algorithm, if one is allowed to hardcode $n^{\sparam}$ bits of information to each machine (which do not depend on the input graph $G$). Our result improves over the state-of-the-art deterministic $O(\log \Delta+\log\log n)$-round algorithm for this problem by \cite{CDP20b} which works in standard low-space MPC model (i.e., with polynomial local computation). This also matches the \emph{randomized} complexity of the problem as given by Chang et al. \cite{CFGUZ19}.


\paragraph{Low-Space \MPC with heavy local computation.}
As noted in previous works, e.g., Andoni et al. \cite{ANOY14}, the main focus of the low-space \MPC model is on the \emph{information-theoretic} question of understanding the round complexity within sublinear space restrictions \emph{(i.e., even with unbounded computation per machine)}. This point of view might provide an explanation for the inconsistency and ambiguity concerning the explicit restrictions on local computation in the low-space \MPC model. Many of the prior work explicitly \emph{allow} for an \emph{unlimited local computation}, e.g., \cite{ANOY14,ASW19,BBDFHKU19,BDELM19,Ghaffari0T20}. Other works only recommend having a polynomial time computation \cite{GKU19,GU19}, and some explicitly restrict the local computation to be polynomial \cite{CFGUZ19,GhaffariJN20}. In this work we take the distributed perspective on the \MPC model by adopting the standard assumption in which local computation comes for \emph{free}, as assumed in all the classical distributed models, \local, \congest\ and \congc. The main motivation for such an assumption is that it decouples \emph{communication} from \emph{computation}.
Our results may indicate that allowing heavy local computation might provide an advantage in the context of distributed and parallel derandomization. 

\subsection{Key Techniques}
Our approach is based on a derandomization of the CLP algorithm using a \emph{pseudorandom generator} \cite{Vadhan12}. As a starting point, we assume that the maximum degree $\Delta$ is in the range $\Delta \in [\poly\log n, n^{\sparam/c}]$ for a sufficiently large constant $c$. The upper bound degree assumption is made possible by employing first a recursive graph partitioning, inspired by \cite{CDP20b}, that uses bounded-independence hash functions to break the problem into several independent instances with lower degree of at most $n^{\sparam/c}$. This allows us to allocate a machine $M_v$ for every node $v$ in the graph, and store on that machine the $O(1)$-radius ball of $v$ in $G$. Since most of the CLP procedures are based inspecting the $O(1)$-radius balls, that would be very useful. To handle small (polylogarithmic) degrees, we employ a derandomization of the state-of-the-art $(deg+1)$-coloring algorithm of Barenboim et al. \cite{BEPS16}. Assuming that $\Delta= \Omega(\poly\log n)$ provides us a more convenient start point for the CLP derandomization, since in this degree regime, all the local randomized CLP procedures succeed with high probability, of $1-1/n^{c'}$, for any desired constant $c'$.
%

The common derandomization approach in all-to-all communication models is based on a combination of obtaining a small search space (i.e., using short random seed) and the method of conditional expectations \cite{CPS17,Luby93}.  The main obstacle in derandomizing the CLP algorithm is that it applies local\footnote{By local we mean that these procedures are part of the local computation of the nodes.} randomized procedures that seem to require almost full independence. It is thus unclear how derandomize them using the standard bounded-independence tools of e.g., hash functions. For example, one of the key procedures for coloring dense regions in the graph (denoted as \emph{almost-cliques}) is based on a \emph{randomized permutation} of the clique' nodes. It is unclear how simulate such a permutation using a small seed and in polynomial time computation. We therefore sacrifice the latter requirement, by allowing heavy local computation.

A \emph{pseudorandom generator} \cite{NisanW88,Vadhan12} is a function that gets a short random seed and expands it to into a long one which is \emph{indistinguishable} from a random seed of the same length for a given class of algorithms. Informally, a PRG function $\mathcal{G}:\{0,1\}^a \to \{0,1\}^b$, where $a\ll b$, is said to $\epsilon$-fool a given class of randomized algorithms $\mathcal{C}$ that uses $b$ random coins as part of their input, if the following holds for every algorithm $C \in \mathcal{C}$: the success probability of $C$ under $b$ pseudorandom coins $\mathcal{G}(X)$, where $X$ is a vector of $a$ random coins, is within $\pm \epsilon$ of the success probability of $C$ when using $b$ truly random coins. Explicit PRG constructions with small seed length have been provided for a collection of Boolean formulas \cite{GopalanMRTV12}, branching program with bounded widths \cite{braverman2014pseudorandom,MekaRT19}, and small depth circuits \cite{DeETT10,NisanW88}.
Unfortunately none of these computational settings fits the local randomized computation of the CLP procedures that we wish to derandomize. A useful property, however, of the CLP procedures is that they run in polynomial time.

Our derandomization is based on a brute-force construction of PRG functions that can $\epsilon$-fool the family of all polynomial time computation using a seed length of $O(\log n)$ bits, for $\epsilon=1/n^c$. The drawback of these PRGs is that they are non-explicit (though can be found by an expensive brute-force computation), and require space which is exponential in the seed length to specify. This, in particular, implies that even if we relax the local computation constraint, in order to fit the space limitations of the low-space \MPC model, we must introduce an additive \emph{sublinear} error of $1/n^\alpha$ for some small constant $\alpha$ that depends on the low-space exponent $\sparam$. In other words, one \emph{can} simulate the CLP procedures using a PRG which fits in machines' local space, but this PRG requires (i) local computation which is exponential in the local space bound, and (ii) a weakened success guarantee to $1-1/n^\alpha$ for a small constant $\alpha \in (0,1)$. We next explain how to handle this larger probability of errors.

\paragraph{Handing sublinear errors.} The increase in the error using small seeds creates complications in several CLP procedures, for the following reason. The CLP procedures are highly sensitive to the \emph{order} in which the nodes get colored. In particular for certain classes of nodes, the analysis is based on showing the each coloring step did not color \emph{too many} neighbors of a given node, while at the same time, colored a sufficiently many neighbors of that node. In other words, a given node (or a cluster of nodes) is happy at the end of a given randomized procedure if its coloring status
satisfies a given (in many cases delicate and non-monotone) invariant that also depends on the coloring status of its neighbors.

It is non-trivial to derandomize such procedures when suffering from a sublinear error. To see this, assume that the machines can compute a PRG that $\epsilon$-fools the CLP local procedures with $\epsilon=1/\sqrt{n}$ with a random seed of $\ell=O(\log n)$ bits. Using standard voting on the $2^{\ell}$ possible seeds, the machines can compute the seed $Z^*$ which maximizes the number of happy nodes. Due to the error of $\epsilon$, this implies that all but $\sqrt{n}$ of the nodes are happy. This appears to be quite a large amount of progress. Indeed, at first glance it may seem that one can complete the computation with only one more recursive step over the remaining $\sqrt{n}$ unhappy nodes. The key complication of this approach is that it might now be impossible to make the remaining $\sqrt{n}$ nodes happy under the current color selection to their happy neighbors, since this may have destroyed some necessary properties for the coloring algorithm.  Furthermore, if one now starts canceling the colors already assigned to happy neighbors, it might create long cancellation chains, ending up with uncoloring all the nodes.

We overcome this impasse using several different approaches, depending on the precise properties of the CLP procedure and of the node class on which it is applied. For example, for one derandomization procedure (Section \ref{sec:slack}), we combine PRGs with bounded-independence hash functions. Informally, the latter are used to partition nodes into groups, to which we apply our PRG in turn, in such a way that error from the PRG can only cause damage within each group, and any of the remaining groups still have the necessary properties to make all nodes happy. In another procedure (Section \ref{sec:no-slack}), where we apply the PRG to \emph{clusters} of nodes, we extend the happiness property of a cluster $S$ to also include conditions on neighboring clusters as well as well as $S$ itself. By carefully choosing these conditions, we will then see that we \emph{can} safely uncolor clusters that do not satisfy their self-related conditions, without violating the necessary conditions of their neighbors. In this way we avoid causing chains of cancellations.

In Section \ref{sec:PRG}, we provide the formal PRG definitions, and describe the general (partial) derandomization in more details. 
	

\section{Algorithm Descriptions}\label{sec:high}

\subsection{Terminology and a Quick Exposition of the CLP Algorithm}\label{sec:CLP-high}
In the description below, we focus on the main randomized part (a.k.a the \emph{pre-shattering }part) of the CLP algorithm \cite{chang2020distributed}, that runs in $O(\log^* \Delta)$ rounds. We start by providing useful definitions, originally introduced by Harris, Schneider and Su \cite{harris2016distributed}.

\begin{definition}
 For an $\epsilon \in (0,1)$, an edge $e=(u,v)$ is called an $\epsilon$-\emph{friend} if $|N(u) \cap N(v)|\geq (1-\epsilon)\Delta$. The endpoints of an $\epsilon$-friend edge are called $\epsilon$-\emph{friends}. A node $v$ is denoted as $\epsilon$-\emph{dense} if $v$ has at least $(1-\epsilon)\Delta$ many $\epsilon$-friends, otherwise it is $\epsilon$-\emph{sparse}. An $\epsilon$-\emph{almost clique} is a connected component of the subgraph induced by the $\epsilon$-dense nodes and their incident $\epsilon$-friend edges.
\end{definition}

The next lemma summarizes the key properties of the almost-cliques.
Throughout, assume that $\epsilon<1/5$ and let $V^{d}_{\epsilon}, V^{s}_{\epsilon}$ be the subsets of $\epsilon$-dense ($\epsilon$-sparse) nodes.
\begin{lemma}[Lemma 3.1 of \cite{chang2020distributed}]\label{lem:blocks}
For every $\epsilon$-almost clique $C$ and every $v \in C$ it holds:
\begin{itemize}
\item $|(N(v) \cap V^{d}_{\epsilon})\setminus C|\leq \epsilon \Delta$ (i.e., small external degree w.r.t $\epsilon$-dense nodes).
\item $|C \setminus (N(v) \cup \{v\})|<3\epsilon\Delta$ (small antidegree).
\item $|C|\leq (1+3\epsilon)\Delta$ (small size).
\item $\dist_G(u,v)\leq 2$ for each $u,v \in C$.
\end{itemize}
\end{lemma}
The CLP algorithm starts by applying a $O(1)$-round randomized procedure that colors a subset of the nodes in a way that generates for the remaining uncolored nodes a \emph{slack} in their number colors. Formally the slack of a node is measured by the difference between the number of colors available in its palette and its uncolored degree.
Our focus will be coloring on the uncolored nodes, denoted by $V^*$. The procedure is based on computing a hierarchy of $\epsilon$-almost cliques for a sequence of increasing $\epsilon$ values $\epsilon_1 <\ldots< \epsilon_\ell$. The hierarchy partitions $V^*$ into $\ell=O(\log\log \Delta)$ \emph{layers} as follows. Layer $1$ is defined by $V_1=V^* \cap V^d_{\epsilon_1}$ and $V_i=V^*\cap (V^d_{\epsilon_i}\setminus V^d_{\epsilon_{i-1}})$ for every $i \in \{2,\ldots, \ell\}$. Letting $V_{sp}=V^* \cap V^s_{\ell}$, we have that $(V_1,\ldots, V_\ell, V_{sp})$ is a partition of $V^*$. The nodes of $V_i$ are denoted as layer-$i$ nodes, these nodes are $\epsilon_i$-dense on the one hand, and also $\epsilon_{i-1}$-sparse on the other hand. The nodes of $V_{sp}$ are denoted as \emph{sparse} nodes.

\noindent \textbf{Blocks:} The layer-$i$ nodes $V_i$ are further partitioned into \emph{blocks}, which refer to a set of layer-$i$ nodes in a given almost-clique. Letting $(C_1,\ldots, C_k)$ be the list of $\epsilon_i$-almost cliques, define the block $B_j=C_j \cap V_i$. The block-list $(B_1, \ldots, B_k)$ is a partition of $V_i$. A block $B_j \subseteq V_i$ is called a \emph{layer}-$i$ \emph{block}. The blocks are classified into three types based on their size: \emph{small, medium and large}. A layer-$i$ block $B$ is \emph{large-eligible} if $|B| \geq \Delta/\log(1/\epsilon_i)$. The division into the three types depend on the relations between the blocks which can be captured by a rooted tree $\mathcal{T}$. For $i<i'$, a layer-$i$ block $B$ is a descendant of a layer-$i'$ block $B'$ if both are subsets of the same $\epsilon_{i'}$-almost clique. The root of the tree   $\mathcal{T}$ is the set $V_{sp}$ of the sparse nodes. The set of \emph{large} blocks is a maximal set of large-eligible and independent blocks (i.e., which are not ancestors or descendants of each other) which prioritizes by size and breaking ties by layer. \emph{Medium} blocks are large-eligible blocks which are not large, and the remaining blocks are \emph{small}. Let $V^S_i,V^M_i$ and $V^L_i$ be the set of layer-$i$ nodes in a layer-$i$ small (medium and large, resp.) blocks. For each $X \in \{S,M,L\}$, let $V^X_{2+}=\bigcup_{i=2}^{\ell}V^X_i$.  The nodes $V^* \setminus V_{sp}$ are colored in six stages according to the order
$$(V^S_{2+},V^S_1, V^M_{2+}, V^M_1, V^L_{2+}, V^L_1)~.$$
This ordering ensures that when a given node is considered to be colored, it has sufficiently many remaining colors in its palette. At the end of these six stages, there will be a small subset $U \subset V^*\setminus V_{sp}$ of uncolored nodes. The sets $V_{sp} \cup U$ will be colored later on efficiently within $O(\log^*\Delta)$ rounds.
The main benefit of defining the six classes is in providing a sufficient amount of slack when considering a given node for coloring.
\begin{lemma}\label{lem:slack-small-med}[Lemma 3.3 of \cite{chang2020distributed}]
For each layer $i \in [1,\ell]$, the following are true:
\begin{itemize}
\item $\forall v \in V^S_i$ with $|N(v) \cap V^*|\geq \Delta/3$, we have
$|N(v) \cap (V^M_{2+} \cup V^M_1 \cup V^L_{2+} \cup V^L_1 \cup V_{sp})|\geq \Delta/4~.$

\item For each $v \in V^M_i$, we have $|N(v)\cap (V^L_{2+} \cup V^1_L \cup V_{sp})|\geq \Delta/(2\log(1/\epsilon_i))$.

\end{itemize}
\end{lemma}
Since the nodes in small and medium blocks have many neighbors in the other sets, when coloring these nodes we enjoy their excess in the number of colors (restricted to their neighbors in the given class). In what follows, we provide a derandomization scheme for each of the randomized procedures applied in the CLP algorithm. For convenience, the pseudocodes of the CLP procedures are provided in Section \ref{sec:CLP-code}. We also stick, in general, to the same notation used by the CLP algorithms.

\subsection{High-Level Description of our \MPC\ Algorithm}
Throughout, a degree bound $\Delta'$ is said to be \emph{medium} if $\Delta'=O(n^{\beta})$ for some constant $\beta$ sufficiently smaller than $\sparam$. In addition, $\Delta'$ is \emph{low} if it is polylogarithmic.
\\
\begin{mdframed}[hidealllines=false,backgroundcolor=gray!30]
\begin{itemize}
\item \textbf{Step (S1): Degree Reduction via Recursive Partitioning}
\begin{itemize}
\item Partition the nodes and their palettes into medium-degree instances using bounded independent hash functions.
\end{itemize}
\item \textbf{Step (S2): CLP Derandomization}
\begin{itemize}
\item Apply a derandomization of the CLP algorithm on each of the medium-degree instances, where all instances of the same recursive level are handled simultaneously in parallel. This reduces the uncolored degree to polylogarithmic.
\end{itemize}
\end{itemize}
\end{mdframed}

\noindent \textbf{Step 1:} In the first step, we reduce the problem to graphs with maximum degree $n^{\beta}$ for any desired constant $\beta \in (0,1)$. This step takes $O(1)$ number of rounds, using bounded-independence hash functions. Our deterministic graph partitioning has the same properties as the randomized partitioning of Chang et al. \cite{CFGUZ19}. 

\noindent \textbf{Step 2:} The second step of the algorithm assumes that $\Delta=O(n^{\beta})$ .
This allows one to store a constant-radius ball of a node on a given machine. Similarly to the CLP algorithm, the derandomization has three main parts: (i) initial coloring (which generates the initial excess in colors) (ii) dense coloring (e.g., coloring nodes with almost-clique neighborhoods) and (iii) coloring bidding which colors nodes with excess colors. Parts (i) and (ii) are derandomized within $O(1)$ number of rounds, and part (iii) is derandomized in $O(\log^*\Delta)$ rounds. In our algorithm, we apply only a partial implementation of procedure (ii), as our goal is to reduce the uncolored degree to a polylogarithmic bound. The coloring of the dense nodes is completed within $O(\log\log\log n)$ rounds by derandomizing the $(\deg+1)$ list coloring algorithm by Barenboim et al. \cite{BEPS16}.


\paragraph{Road-map.} In Sections \ref{sec:PRG} and \ref{sec:bounded-hash} we present our derandomization tools of PRG and bounded-independence hash functions. Note that Section \ref{sec:PRG} introduces notations that will be used throughout our algorithms.  In Section \ref{sec:low-deg}, we first provide a deterministic $(deg+1)$ coloring algorithm for graphs with polylogarithmic degrees. We therefore assume from now on that $\Delta\geq \log^c n$ for a sufficiently large constant $c$. In Section \ref{sec:red-low-deg} we describe the first step of our coloring algorithm, where we apply a recursive partitioning which results in medium degree coloring instances. The derandomization of CLP of Step 2 spans over Sections \ref{sec:slack}, \ref{sec:small-med-col}, and \ref{sec:no-slack}. In Section \ref{sec:slack} we provide a derandomization of the $\OneShotColoring$ procedure for generating the \emph{initial} color excess for every node as a function its neighborhood sparsity.
Then we turn to consider the coloring of the dense vertices $V \setminus V_{sp}$. Recall that these nodes are partitioned into the classes $(V^S_{2+},V^S_1, V^M_{2+}, V^M_1, V^L_{2+}, V^L_1)$.
In Section \ref{sec:small-med-col}, we provide a derandomization of the CLP procedures for coloring the dense nodes in small and medium blocks $V^S_{2+},V^S_1, V^M_{2+}, V^M_1$. Section \ref{sec:no-slack} considers the remaining dense nodes in the large blocks $V^L_{2+}, V^L_1$. Our derandomization reduces the uncolored degrees of the nodes in $V^L_{2+}$ as a function of their sparsity. In addition, it reduces the degrees of the uncolored nodes in $V^L_1$ to polylogarithmic. The coloring of the remaining $V^L_1$ can be then completed in $O(\log\log\log n)$ rounds.
Section \ref{sec:sparse} handles the remaining uncolored vertices in layer $\geq 2$, as well as the sparse vertices $V_{sp}$.

%
%
\vspace{-5pt}\subsection{Pseudorandom Generators and Derandomization}\label{sec:PRG}
We will now formally define pseudorandom generators (PRGs). A PRG is a function that gets a short random seed and expands it to a long one which is indistinguishable from a random seed of the same length for a given class of algorithms. We will use the following definitions from \cite{Vadhan12}: in the latter, $U_k$ denotes the uniform distribution on $\{0,1\}^k$.
\begin{definition}[\textbf{Computational Indistinguishability}, Definition 7.1 in \cite{Vadhan12}]
Random variables $X$ and $Y$ taking values in $\{0,1\}^m$ are $(t,\epsilon)$ \emph{indistinguishable} if for every nonuniform algorithm $T$ running in time at most $t$, we have
$|\Pr[T(X)=1]-\Pr[T(Y)=1]| \le \epsilon$.
\end{definition}
\begin{definition}[\textbf{PRG}, Definition 7.3 in \cite{Vadhan12}]
A deterministic function $\mathcal{G}:\{0,1\}^d \to \{0,1\}^m$ is an $(t,\epsilon)$ pseudorandom generator (PRG) if:
(1) $d\leq m$ and (2) $\mathcal{G}(U_d)$ and $U_m$ are $(t,\epsilon)$ indistinguishable.
\end{definition}

A simple counting argument (given in, e.g., \cite{Vadhan12}) shows that there must \emph{exist} PRGs with short seeds:

\begin{proposition}[Proposition 7.8 in \cite{Vadhan12}]
\label{prop:perfect-PRG}
For all $m \in \mathbb{N}$ and $\epsilon>0$, there exists a (non-explicit) $(m,\epsilon)$ PRG $\mathcal{G}:\{0,1\}^d \to \{0,1\}^m$ with seed length $d=O(\log m+\log 1/\epsilon)$.
\end{proposition}
The next lemma follows by a brute-force PRG construction from \cite{CDP21a}.

\begin{lemma}
\label{lem:prg-alg}
For all $m \in \mathbb{N}$ and $\epsilon>0$, there exists an algorithm for computing the $(m,\epsilon)$ PRG of Proposition \ref{prop:perfect-PRG} with seed length $d=O(\log m+\log 1/\epsilon)$, in time $\exp(\poly(m/\eps))$ and space $\poly(m/\eps)$.
\end{lemma}
%

\paragraph{Derandomization with PRG}
A randomized \local\ algorithm $\cA$ is said to be \emph{nice} if the local computation, per round, performed at each node is polynomial in $\Delta$. All the randomized \local\ procedures that we derandomize with the PRG framework in this paper, will indeed be nice. We will also have the property that $\poly(\Delta)$ bits fit the local space of each machine. To illustrate the technique, assume that $\cA$ is a \emph{two} round randomized algorithm such that after applying $\cA$, each node satisfies a given desired property that depends only on its $1$-radius ball, with high probability of $1-1/n^c$.
It is convenient to view this two-round \local\ algorithm in a way that decouples the randomness from the computation. Specifically, we assume that each node a priori generates its own pool of random coins, and we simulate $\cA$ in the \local\ model by letting each node $v$ first collects its two-hop ball $v \cup N_{G^2}(v)$, as well as, the initial states and the private coins of each of its $2$-hop neighbors. Then, each node $v$
locally applies an algorithm $\cA_v$ on this information. In this view, every algorithm $\cA$ consists of $n$ sub-algorithms $\{\cA_v, v \in V\}$ where each $\cA_v$ is a randomized $\poly(\Delta)$-time algorithm, the randomized decisions made by each node $u$ are consistent with all the algorithms run by its $2$-hop neighbors. We then say that a node $v$ is \emph{happy} if a certain property holds for $\{v\}$. Our goal is to show that using the PRG framework, there is a low-space MPC deterministic algorithm that derandomizes $\cA$ in a way that makes at least $1-1/n^{\alpha}$ fraction of the nodes happy, for some constant $\alpha \in (0,1)$.

The first preprocessing step for the derandomization computes $O(\log \Delta)$-bit identifiers for the nodes such that the identifiers are distinct in each $2$-radius ball. This can be done in $O(\log^* n)$ deterministic MPC rounds \cite{Kuhn09,Linial92}.

\begin{claim}\label{cl:aux}
Given that all the nodes have $O(\log \Delta)$-bit identifiers (unique in each $2$-radius ball), there exists a
low-space \MPC\ deterministic algorithm that causes a collection of at least $(1-1/n^{\alpha})|V|$ nodes to be happy, for some constant $\alpha$ sufficiently smaller than $\sparam$. The round complexity of the algorithm is $O(1)$, and it requires .
\end{claim}
The derandomization is based on two parts. First, we show that there is a weaker variant of algorithm $\cA$ that uses only a shared random seed of $c\cdot\sparam \cdot \log n$ bits for a sufficiently small constant $c \in (0,1]$. This weaker variant suffers an additive error of $\pm 2/n^{\alpha}$ (for some constant $\alpha$ sufficiently smaller than $\sparam$), compared to the fully-randomized algorithm. Then, we derandomize this weaker variant in $O(1)$ \MPC\ rounds in such a way that at least $(1-1/n^{\alpha})|V|$ nodes are happy.

We start with the first step, which is the part that exploits the PRG machinery. Let $t=\poly(\Delta)$ be an upper bound on the local time complexity of all $\{\cA_v, v\in V\}$ algorithms, and let $N=\Delta^c$ be an upper bound on the largest node identifier. The local randomized algorithm $\cA_v$ applied locally at each node $v$ can be represented as a \emph{deterministic} algorithm that gets as input a vector of $N \cdot t$ random coins interpreted as follows: the $i^{th}$ chunk of $t$ coins specifies the random coins for a node with ID $i$ for every $i \in \{1,\ldots, N\}$. Since the algorithm $\cA_v$ runs in time $t$, it is sufficient to specify at most $t$ random bits for each node in $\{v\} \cup N_{G^2}(v)$. The weaker randomized algorithm, denoted by $\cA'_v$, will be given a collection of $N \cdot t$ \emph{pseudorandom} coins obtained by applying a PRG function $\mathcal{G}^*$ on a shared random seed of only $c\sparam\log n$ random coins, for a sufficiently small constant $c \in (0,1]$.

Specifically, by Prop. \ref{prop:perfect-PRG}, there is an $(N \cdot t,\epsilon)$ pseudorandom generator $\mathcal{G}^*:\{0,1\}^d \to \{0,1\}^{Nt}$ with a seed length $d=O(\log (Nt)+\log 1/\epsilon)$. This $\mathcal{G}^*$ function  $\epsilon$-fools the collection of all $t$-time randomized algorithms up to an additive error of $\epsilon$. In particular, it fools the collection of all $\{\cA_v ~\mid~ v \in V\}$ algorithms. We choose the constant in the seed length to be small enough so that the machines will be able to locally compute $\mathcal{G}^*$ within their space limitations. Since the PRG computation consumes $2^{d}\poly(\Delta)$ space, we can support an additive error of $\epsilon=1/(2n^{\alpha})$ for some small constant $\alpha$. This provides a seed of length $d=c\sparam\log n$ for a small constant $c \in (0,1]$. Let $Z \in \{0,1\}^d$ be a random seed of length $d$. We then have that when each node $v$ simulates $\cA_v$ using $\mathcal{G}^*(Z)$ as the source of $N t$ pseudorandom coins, the node $v$ is happy with probability of $1-1/n^c-1/(2n^{\alpha})\leq 1-1/n^{\alpha}$.  It is important to note that since all nodes use the shared random seed $Z$, and since the pseudorandom coins assigned by each local algorithm $\cA'_v$ to node $u \in \{v\} \cup N(v)$ are determined by the $O(\log \Delta)$-bit identifier of $u$, the output of $u$ is consistent by all the algorithms $\cA_w$ for every $w \in \{u\} \cup N_{G^2}(u)$.

It remains to show that these weaker algorithms $\{\cA'_v\}$ can be derandomized within $O(1)$ rounds. This is done by allocating a machine $M_v$ for every node $v$ that stores also the $2$-radius ball of $v$ in $G$. \hide{The algorithm also allocates a machine $M_Z$ for every possible seed $Z \in \{0,1\}^d$. } Every machine $M_v$ simulates algorithm $\cA'_v$ under each $Z \in \{0,1\}^d$. Specifically, for each $Z \in \{0,1\}^d$, it simulates $\cA_v$ using $\mathcal{G}^*(Z)$ as the input of random coins. This allows each machine $M_v$ to determine if $v$ is happy under each possible seed $Z$. As there are $o(n^{\sparam})$ seeds, this fits the local space. Using standard sorting procedures, in constant rounds the machines can compute the seed $Z^*$ that maximizes the number of happy nodes. Finally, all machines simulate $\cA'_v$ using the seed $Z^*$, which defines the output of the algorithm. Since the expected number of happy nodes with a random seed $Z \in \{0,1\}$ is at least $(1-1/n^{\alpha})|V|$, the number of happy nodes under the best seed $Z^*$ is at least $(1-1/n^{\alpha})|V|$ as well.

In Section \ref{sec:no-slack}, the procedure is slightly more complicated: there, the definition of a happy \emph{cluster} $S$ will contain both \emph{self-invariants} (properties about $S$) and \emph{neighbor-invariants} (properties about neighboring clusters to $S$). We use the PRG to show that the number of happy clusters is at least a $(1-1/n^{\alpha})$-proportion of the total number of clusters, as above. However, unhappy clusters are then uncolored. By the choice of the happiness properties, we show that clusters that were happy still satisfy their self-invariants, which will be sufficient to allow unhappy uncolored neighbors the chance to become happy even when running the coloring algorithm on the remaining unhappy nodes. The analysis will show that within $O(1/\alpha)$ iterations, all clusters become happy\hide{ up to a constant slack}. These arguments involve introducing some extra slack to the bounds of CLP: specifically, our bounds (e.g., on the desired uncolored degrees) at the end of this derandomization will be larger by a factor of $(1/\alpha)=O(1)$ than those obtained by the randomized CLP procedures.

Throughout the paper we use the parameter $\alpha$ to determine the exponent of the additive error of the pseudorandom algorithms. That is, we will only be considering \local\ randomized algorithms that succeeds at each node with high probability, and using the space limitation of the machines, we will a get a pseudorandom \local\ algorithm that succeeds with probability of $1-1/n^{\alpha}$ for some sufficiently small $\alpha$. Consequently, the derandomization will cause at least a $(1-1/n^{\alpha})$ fraction of the nodes to be happy.

\paragraph{Useful Observations for the CLP Algorithm.}
The above mentioned general derandomization scheme fits well into our setting of derandomizing the CLP algorithm. Specifically, we observe the following useful property for the CLP algorithm. This allows us to work, throughout, with $O(\log \Delta)$-bit identifiers, which is crucial for the derandomization procedure. We show:

\begin{observation}\label{obs:derand-CLP-smallID}
All the randomized procedures of the CLP algorithm are nice, and in addition, these algorithms can be simulated in an analogous manner in the following setting: all nodes are given $O(\log\Delta)$-bit identifiers that are unique within each $O(1)$-radius ball (for any desired constant), and all nodes of the same identifier are given the same set of $\poly(\Delta)$ random coins to simulate their random decisions.
\end{observation}
\begin{proof}
The claim holds immediately for the initial coloring procedure and for the dense coloring procedures as they have a round complexity of $O(1)$ rounds. The coloring bidding procedure of Lemma 2.2 in \cite{chang2020distributed} runs in $O(\log^*\Delta)$ rounds, by Lemma 2.2. the guarantee for each node $v$ holds even if the random bits generated outside a $O(1)$-radius ball of $v$ are determined adversarially. Since nodes in a given $O(1)$-radius ball have unique IDs, they are given an independent collection of random coins, which is sufficient for the color bidding procedure.
\end{proof}

\begin{observation}
All the randomized procedures of the CLP algorithm are nice, run in $O(1)$ rounds\footnote{The sparse coloring procedure runs in $\ell=O(\log^*\Delta)$ rounds but it is based on $\ell$ applications of a $O(1)$-round procedure, which we will derandomize with PRGs.}, and satisfy desired properties at each node \emph{with high probability} (provided that $\Delta$ is at least polylogarithmic).
\end{observation}
%

\subsection{Bounded-Independence Hash Functions}\label{sec:bounded-hash}
Some of the local randomized procedures considered in this paper requires a more light-weight derandomization scheme which based on bounded-independence hash functions. Like PRGs, these functions can approximate the effect of random choices using only a small seed (which, additionally, can be computed efficiently in polynomial time, rather than in exponential time as in the PRG setting). The major benefit of bounded-independence hash functions over PRGs is that they do not incur \emph{error}, as the PRGs we use do. That is, they provide \emph{exactly} the same bounds as true randomness for analyses that only require independence between a bounded number of random choices.

The families of hash functions we require are specified as follows:
\vspace{-3pt}
\begin{definition}
For $N, L, k \in \nat$ such that $k \le N$, a family of functions $\mathcal{H} = \{h : [N] \rightarrow [L]\}$ is \emph{$k$-wise independent} if for all distinct $x_1, \dots, x_k \in [N]$, the random variables $h(x_1), \dots, h(x_k)$ are independent and uniformly distributed in $[L]$ when $h$ is chosen uniformly at random from~$\mathcal{H}$.
\end{definition}
\noindent We will use the following well-known lemma (see, e.g., \cite[Corollary~3.34]{Vadhan12}).
\vspace{-3pt}
\begin{lemma}
\label{lem:hash}
For every $a$, $b$, $k$, there is a family of $k$-wise independent hash functions $\mathcal{H} = \{h : \{0,1\}^a \rightarrow \{0,1\}^b\}$ such that choosing a random function from $\mathcal{H}$ takes $k \cdot \max\{a,b\}$ random bits, and evaluating a function from $\mathcal{H}$ takes time $poly(a,b,k)$.
\end{lemma}

\subsection{The Method of Conditional Expectations}\label{sec:condexp}
To agree on seeds specifying a particular hash function, we make use of a distributed implementation of the classical \emph{method of conditional expectations}, as employed in \cite{CDP20a,CDP20b}. The properties of this method can be stated as follows:

Assume that, over the choice of a random hash function $h \in \mathcal{H}$, the expectation of some objective function (which is a sum of functions $q_x$ calculable by individual machines) is at least some value $Q$. That is,
\begin{align*}
	\mathbf{E}_{h \in \mathcal{H}} \big[{q(h):=\sum_{\text{machines } x} q_x(h)}\big]
	&\ge
	Q
	\enspace.
\end{align*}

Then, if $|\mathcal{H}| = poly(n)$, i.e., hash functions from $\mathcal H$ can be specified using $O(\log n)$ bits, there is an $O(1)$-round deterministic low-space \MPC algorithm allowing all machines to agree on some specific $h^* \in \mathcal{H}$ with $q(h^*) \ge Q$.

The method with which we fix seeds for our PRGs can also be thought of as a special case of this implementation of the method of conditional expectations; the difference is that the seeds for our PRGs are short enough that the entire seed space fits in a machine's memory and can be searched at once, whereas our hash function seeds are a constant-factor longer and the seed space must be searched in $O(1)$ iterations.

%

\section{Reducing to Medium-Degree Instances}\label{sec:red-low-deg}

 In this section, we provide a deterministic (recursive) degree reduction procedure, and show the following:

\begin{lemma}\label{lem:low-deg-reduc}
	Assume that for every $n$, there is an $O(\log\log\log n)$-round low-space \MPC\ algorithm $\cA$ for computing $(\Delta+1)$ list coloring in a graph $G$ provided that (i) $\Delta=n^{O(\zeta)}$ for a constant $\zeta$ sufficiently smaller than $\sparam$, and (ii) each vertex has a palette of size $\min\{\deg_G(v), \Delta-\Delta^{3/5}\}$. Then there is an $O(\log\log\log n)$-round $(\Delta+1)$ list coloring algorithm for \emph{any} $n$-vertex graph $G$.
\end{lemma}

We do so by using an extension of a method introduced in \cite{CDP20b}, which deterministically reduces an instance of coloring with high degree to a collection of instances of lower degree, of which many can be solved in parallel. We run the following algorithm up to a recursion depth of $\frac{\log_n \Delta}{\zeta} - 23$ (note that we have control of $\zeta$, and set it so that this is an integer) - upon reaching that recursion depth we apply algorithm $\cA$ in place of the recursive call to \textsc{LowSpaceColorReduce}. Throughout we use families of $O(1)$-wise independent hash functions. See Section \ref{sec:bounded-hash} for definitions.

\begin{algorithm}[H]
	\caption{\textsc{LowSpaceColorReduce}$(G)$}
	\label{alg:LSColorReduce}
	\begin{algorithmic}
		\State $G_1, \dots, G_{n^{\zeta}} \gets \textsc{LowSpacePartition}(G)$.
		\State For each $i = 1, \dots,n^{\zeta}-1$, perform \textsc{LowSpaceColorReduce}$(G_i)$ in parallel.
		\State Update color palettes of $G_{n^{\zeta}}$, perform \textsc{LowSpaceColorReduce}$(G_{n^{\zeta}})$.
	\end{algorithmic}
\end{algorithm}

This algorithm employs a partitioning procedure to divide the input instance into \emph{bins}, which are then solved recursively:

\begin{algorithm}[H]
	\caption{\textsc{LowSpacePartition}$(G)$}
	\label{alg:LSPartition}
	\begin{algorithmic}
		\State Let hash function $h_1:[poly(n)]\rightarrow [n^{\zeta}]$ map each node $v\notin V_0$ to a bin $h_1(v) \in [n^{\zeta}]$.
		\State Let hash function $h_2:[poly(n)]\rightarrow [n^{\zeta}-1]$ map colors $\gamma$ to a bin $h_2(\gamma) \in [n^{\zeta}-1]$.
		\State Let $G_1,\dots,G_{n^{\zeta}}$ be the graphs induced by bins $1,\dots,n^{\zeta}$ respectively.
		\State Restrict palettes of nodes in $G_1,\dots,G_{n^{\zeta}-1}$ to colors assigned by $h_2$ to corresponding bins.
		\State Return $G_1,\dots,G_{n^{\zeta}}$.
	\end{algorithmic}
\end{algorithm}

(If the ranges of the hash functions are not powers of $2$, we can instead map to a sufficiently large (but $poly(n)$) power of $2$, and then map intervals of this range to $[n^{\zeta}]$ or $[n^{\zeta}-1]$ as equally as possible. As in \cite{CDP20a,CDP20b}, we incur some error in the distributions of random hash function outputs, but can easily ensure that it is negligibly small, e.g. $n^{-3}$.)

The analysis of this procedure is similar to \cite{CDP20b}, and deferred to Appendix \ref{sec:low-deg-full}.

\def\APPENDLOWDEGCOMP{
\subsection{Complete Analysis for Section \ref{sec:red-low-deg} (Proof of Lemma \ref{lem:low-deg-reduc})}\label{sec:low-deg-full}
We consider what happens for a particular fixed node $v$ during one recursive call to \textsc{LowSpacePartition}$(G)$. Let $d(v)$ and $p(v)$ denote the degree of $v$ within its bin, and the size of its palette respectively. Let $d'(v)$ be $v$'s new degree within its newly partitioned bin, and let $p'(v)$ be its new palette size: the number of colors in $v$'s old palette mapped to its new bin if $h_1(v) =n^{\zeta}$, and $p(v)-(d(v)-d'(v))$ otherwise. This latter expression comes from the fact that $v$'s new palette consists of its old palette ($p(v)$ colors) with colors used by neighbors not in bin $n^\zeta$ removed (at most $d(v)-d'(v)$ colors).

We would like to  define \emph{good} nodes whose behavior does not differ too much from what we would expect from a random partition (and \emph{bad} nodes for whom the opposite is true). However, one of the restrictions of the low-space regime is that we cannot store a node's palette, or all of its adjacent edges, on a single machine, and therefore machines could not determine locally whether a node is good or bad. So, we will instead divide nodes' neighbors and palettes into smaller sets which \emph{do} fit onto single machines, and define a notion of a \emph{machine} being good.

Specifically, for each node $v$, we will create a set $M_N^v$ of machines which will be responsible for $v$'s neighbors. We partition $v$'s neighbor set among machines in $M_N^v$, with each machine receiving $n^{22\zeta}$ neighbors (except possibly one machine which receives the remainder).

Similarly, for each node $v\notin G_{n^{\zeta}}$, we create a set $M_C^v$ of machines responsible for the colors in its palette, and partition the palette among the machines, with each machine (other than one for the remainder) receiving $n^{22\zeta}$ colors.

For consistency of notation, for a machine $x\in M_N^v$, we define $d(x)$ to be the number of neighbors it receives, and $d'(x)$ to be the number of such neighbors which are hashed to the same bin as $v$.
Similarly, for a machine $x\in M_C^v$ holding colors for a node $v$, we define $p(x)$ to be the number of colors received, and $p'(x)$ to be the number of such colors hashed by $h_2$ to the same bin as $v$.

\begin{definition}[\textbf{Good and bad machines}]\mbox{}\label{def:LSgood-and-bad}
	\begin{itemize}
		\item A machine $x\in M_N^v$ is \emph{good} if $|d'(x) - \frac{d(x)}{n^\zeta } | \le n^{11\zeta}$.
		\item A machine $x\in M_C^v$ is \emph{good} if $|p'(x) - \frac{p(x)}{n^\zeta - 1}|\le n^{11\zeta}$.
		\item Machines are \emph{bad} if they are not \emph{good}.
	\end{itemize}
\end{definition}

Our next step is to show that, under a \emph{random} choice of hash functions during partitioning, with high probability there are no bad machines. We make use of the \emph{concentration bound} for sums of bounded-independence variables:

\begin{lemma}[Lemma 2.3 of \cite{BR94}]
\label{lem:conc}
Let $\cj \ge 4$ be an even integer. Suppose $Z_1, \dots, Z_t$ are $\cj$-wise independent random variables taking values in $[0,1]$. Let $Z = Z_1 + \dots + Z_t$ and $\mu = \Exp{Z}$, and let $\lambda > 0$. Then,
\begin{align*}
\Prob{|Z - \mu| \ge \lambda} &\le 8 \left(\frac{\cj \mu + \cj^2}{\lambda^2}\right)^{\cj/2}
\enspace.
\end{align*}
\end{lemma}

\subsection{Good and Bad Machines for \emph{Random} Hash Functions}
\label{subsec:LS-bad-and-good-for-random-hash-function}

We apply the bounded-independence concentration bound to show that node degrees and palette sizes within bins to not differ too much from their expectation:

\begin{lemma}\label{lem:LScoloreduce1}
	For each machine $x \in M_N^v$,  $|d'(x) - \frac{d(x)}{n^\zeta } | \le n^{11\zeta}$ with probability at least $1-n^{-2}$.
\end{lemma}

\begin{proof}
	We apply Lemma \ref{lem:conc} with $Z_1, \dots, Z_{d(x)}$ as indicator random variables for the events that each neighbor held is hashed to $v$'s bin. These variables are $c$-wise independent and each have expectation $n^{-\zeta}$. Since we assume $c$ to be a sufficiently high constant, by Lemma \ref{lem:conc}, we obtain
\[\Prob{|Z-d(x) n^{-\zeta}| \ge n^{11\zeta}} \le 8 \left(\frac{\cj d(x) n^{-\zeta} + \cj^2}{n^{22\zeta}}\right)^{\cj/2}\le 8 \left(\cj n^{-\zeta} + \cj^2 n^{-22\zeta}\right)^{\cj/2} \le n^{-2}
.\]
\end{proof}

\begin{lemma}\label{lem:LScoloreduce2}
	For each machine $x \in M_C^v$, $|p'(x) - \frac{p(x)}{n^\zeta - 1}|\le n^{11\zeta}$ with probability at least $1-n^{-2}$.
\end{lemma}

\begin{proof}
	We apply Lemma \ref{lem:conc}, with $Z_1,\dots,Z_{p(x)}$ as indicator random variables for the events that each color held is hashed to $v$'s bin. These variables are $c$-wise independent and each have expectation $\frac{1}{n^{\zeta}-1}$. Since we assume $c$ to be a sufficiently high constant, by Lemma \ref{lem:conc}, we obtain
	\[\Prob{|Z-\frac{p(x)}{n^\zeta - 1}| \ge n^{11\zeta}} \le 8 \left(\frac{\cj \frac{p(x)}{n^\zeta - 1}+ \cj^2}{n^{22\zeta}}\right)^{\cj/2}
	\le
	8 \left( \frac{\cj }{n^\zeta - 1}+ \cj^2n^{-22\zeta}\right)^{\cj/2}
	\le n^{-2}.\]	
\end{proof}

We will now define a cost function for pairs of hash functions that we would like to minimize. This time our cost function is simpler, since we can weight all of our bad events equally, and with high probability none of them occur.

We define the \emph{cost} function $q(h_1,h_2)$ of a pair of hash functions $h_1 \in \mathcal{H}_1$ and $h_2 \in \mathcal{H}_2$, as follows:
\begin{align}
	\label{def:LSquality}
	q(h_1,h_2)
	&=
	|\{\text{bad machines under }h_1,h_2\}|
	\enspace.
\end{align}

We can then bound the cost of random hash function pairs:

\begin{lemma}\label{lemma:LSexpected-quality}
	The expected cost of a random hash function pair from $\mathcal{H}_1 \times \mathcal{H}_2$ is less than $1$.
\end{lemma}

\begin{proof}
	By Lemma \ref{lem:LScoloreduce1} and Lemma \ref{lem:LScoloreduce2}, any machine is bad with probability at most $n^{-2}$. The number of machines we require is $O(n+ \frac{m}{n^{22\zeta}})=o(n^2)$. Therefore,
	\begin{align*}
		\Exp{q(h_1,h_2)}
		&=
		\Exp{|\{\text{bad machines}\}|}
		\le
		n^{-2} \cdot o(n^2)
		<
		1
		\enspace.
	\end{align*}
\end{proof}

Applying the method of conditional expectations therefore gives the following:

\begin{lemma}\label{lem:LSdeterministic-hashing}
	In $O(1)$ \MPC rounds with $O(n^{22\zeta})$ local space per  machine and $O(m+n)$ total space, one can select hash functions $h_1 \in \mathcal{H}_1$, $h_2 \in \mathcal{H}_2$ such that in a single call to \textsc{LowSpacePartition},
	\begin{itemize}
		\item for any node $v$, $|d'(v) - \frac{d(v)}{n^\zeta } | \le d(v) n^{-11\zeta} + n^{11\zeta}$, and
		\item for any node $v\notin G_{n^{\zeta}}$, $|p'(v) - \frac{p(v)}{n^\zeta -1}| \le p(v) n^{-11\zeta} + n^{11\zeta}$.
	\end{itemize}
\end{lemma}

\begin{proof}
	By the method of conditional expectations (using, as local functions $q_x$, the indicator variables for the event that machine $x$ is bad), we can select a hash function pair $h_1,h_2$ with cost $q(h_1,h_2)<1$; since cost must be a integer, this means that $q(h_1,h_2)=0$, i.e. all machines are good.
	
	By Lemma \ref{def:LSgood-and-bad}, for any node $v$, we therefore have:
	\begin{align*}
		|d'(v) - \frac{d(v)}{n^\zeta } |
		&=
		\sum_{x \in M_N^v} |d'(x) - \frac{d(x)}{n^\zeta } |
		\le
		\sum_{x \in M_N^v} n^{11\zeta}
		\le
		\left(\frac{d(v)}{n^{22\zeta}}+1\right)n^{11\zeta}
		\le d(v) n^{-11\zeta} + n^{11\zeta}
		\enspace.
	\end{align*}
	
This satisfies the first part of the lemma. Similarly, for any node $v\notin  G_{n^{\zeta}}$, we also have:
	\begin{align*}
		|p'(v) - \frac{p(v)}{n^\zeta -1} |
&=
\sum_{x \in M_N^v} |p'(x) - \frac{p(x)}{n^\zeta } |
\le
\sum_{x \in M_N^v} n^{11\zeta}
\le
\left(\frac{p(v)}{n^{22\zeta}}+1\right)n^{11\zeta}
\le p(v) n^{-11\zeta} + n^{11\zeta}
		\enspace.
	\end{align*}
\end{proof}

We can therefore show that whenever we call \textsc{LowSpacePartition}, we maintain our invariant that all nodes have sufficient colors, i.e. $p'(v)>d'(v)$.

\begin{lemma}\label{lem:LSPaletteSize}
	Assume that at the start of our \textsc{LowSpacePartition}$(G)$ call, all nodes $v$ in $G$ have $p(v)>d(v)$. Then, after the call, we have:
		\[d'(v) < d(v) n^{-\zeta}+ d(v) n^{-11\zeta} + n^{11\zeta}.\]	
	
If, additionally,  $p(v)\ge n^{22\zeta}$, then:
	\begin{align*}
		&p'(v) > d'(v),\\
		&p'(v) > p(v) n^{-\zeta} - 2p(v)n^{-11\zeta}.
	\end{align*}
\end{lemma}

\begin{proof}
The first inequality is immediate by Lemma \ref{lem:LSdeterministic-hashing}.

For the second, if $v\notin G_{n^\zeta}$, we have
	\begin{align*}
	p'(v) &\ge \frac{p(v)}{n^\zeta -1} - p(v) n^{-11\zeta} - n^{11\zeta}\\
	&\ge p(v)(n^{-\zeta}+n^{-2\zeta}) - p(v) n^{-11\zeta} - n^{11\zeta}\\
	&\ge p(v)(n^{-\zeta}+n^{-3\zeta})\\
	&> d(v)n^{-\zeta } +  2d(v) n^{-11\zeta} \ge d'(v)
	\enspace.
\end{align*}

Otherwise, if $v\in G_{n^\zeta}$, $p'(v)>d'(v)$ follows since colors are only removed from $v$'s palette when used by its neighbors.

For the third inequality, if $v\notin G_{n^\zeta}$,then from the above calculations we have
	\begin{align*}
	p'(v)
	&\ge p(v)(n^{-\zeta}+n^{-3\zeta})
	> p(v) n^{-\zeta} - p(v)n^{-11\zeta}
	\enspace.
\end{align*}

Otherwise, if $v\in G_{n^\zeta}$, then:

\begin{align*}
	p'(v)
	&\ge p(v)-(d(v)-d'(v))\\
	&\ge p(v)-d(v)+(d(v)n^{-\zeta}-d(v)n^{-11\zeta} - n^{11\zeta})\\
	&> p(v)-p(v)+p(v)n^{-\zeta}-p(v)n^{-11\zeta} - n^{11\zeta}\\
	&\ge  p(v) n^{-\zeta} - 2p(v)n^{-11\zeta}
	\enspace.
\end{align*}	
\end{proof}

We now analyze what happens after $\frac{\log_n \Delta}{\zeta} - 23$ recursion levels:

\begin{lemma}
For a node $v$, let $p^*(v)$ and $d^*(v)$ be the palette size and degree of $v$, within its bin, after a depth of $\frac{\log_n \Delta}{\zeta} - 23$ recursive calls. Then,

\begin{align*}
p^*(v)&\ge n^{23\zeta}-O( n^{13\zeta}),\\
p^*(v)&>d^*(v),\\
d^*(v)&<n^{23\zeta} +O(n^{13\zeta}).
\end{align*}
\end{lemma}

\begin{proof}
Let $j:=\frac{\log_n \Delta}{\zeta} - 23$, let $d_i$ denote the degree of $v$ after $i$ recursive calls, and $p_i$ denote the palette size. We have $d_0 \le \Delta$, and $p_0 = \Delta+1$. From Lemma \ref{lem:LSPaletteSize} we see that we always preserve that every node has more colors in its palette than neighbors.

We first show by induction that for all $i\le j$, $p_i> \Delta n^{-i\zeta} - 2i\Delta n^{-(i+10)\zeta}$. This is clearly true for $i=0$, since $p_i> \Delta$. Assuming the inductive hypothesis for $i$, if $p_i \ge n^{22\zeta}$, then by Lemma \ref{lem:LSPaletteSize} we have
\begin{align*}
p_{i+1} &> p_i n^{-\zeta} - 2 p_i n^{-11\zeta} \\
&> \left(\Delta n^{-i\zeta} - 2i\Delta n^{-(i+10)\zeta}\right)	n^{-\zeta} - 2 \left(\Delta n^{-i\zeta} - 2i\Delta n^{-(i+10)\zeta}\right)n^{-11\zeta}\\
& = \Delta \left(n^{-(i+1)\zeta} - (2i+2)n^{-(i+11)\zeta} +4in^{-(i+21)\zeta}\right) \\
&> \Delta n^{-(i+1)\zeta} - 2(i+1)\Delta n^{-((i+1)+10)\zeta}\enspace.
\end{align*}

So,
\begin{align*}
p^*(v) &\ge \max\{n^{22\zeta}, \Delta n^{-j\zeta} - 2j\Delta n^{-(j+10)\zeta}\} \\
&\ge \max\{n^{22\zeta}, \Delta n^{23\zeta-\log_n \Delta} - 2j\Delta n^{13\zeta-\log_n \Delta}\}\\
&\ge n^{23\zeta} - 2j n^{13\zeta} = n^{23\zeta} - O(n^{13\zeta})
\enspace.
\end{align*}

Similarly, we show by induction that for all $i\le j$, $d_i< \Delta n^{-i\zeta} + 2i\Delta n^{-(i+10)\zeta} + 2n^{11\zeta}$. This is clearly true for $i=0$, since $p_i> \Delta$. Assuming the inductive hypothesis for $i$, then by Lemma \ref{lem:LSPaletteSize} we have
\begin{align*}
	d_{i+1} &< d_{i} n^{-\zeta}+ d_{i}  n^{-11\zeta} + n^{11\zeta}\\
	&< \left(\Delta n^{-i\zeta} + 2i\Delta n^{-(i+10)\zeta}+ 2n^{11\zeta}\right)(n^{-\zeta} + n^{-11\zeta})+n^{11\zeta}\\
	&= \Delta n^{-(i+1)\zeta} + (2i+1)\Delta n^{-(i+11)\zeta}+ 2i\Delta n^{-(i+21)\zeta} + 2n^{11\zeta}(n^{-\zeta} + n^{-11\zeta})+n^{11\zeta}\\
	&< \Delta n^{-(i+1)\zeta} + 2(i+1)\Delta n^{-(i+11)\zeta} + 2n^{11\zeta}
\enspace.
\end{align*}

So,
\begin{align*}
	d^*(v) &< \Delta n^{-j\zeta} + 2j\Delta n^{-(j+10)\zeta} + 2n^{11\zeta} \\
	&= \Delta n^{23\zeta-\log_n \Delta} + 2j\Delta n^{13\zeta-\log_n \Delta} + 2n^{11\zeta} \\
	&=  n^{23\zeta} + 2j n^{13\zeta} + 2n^{11\zeta} =  n^{23\zeta}+O(n^{13\zeta}).
\end{align*}

Finally, $p_i>d_i$ is preserved at every step $i$, so $p^*(v)>d^*(v)$.
\end{proof}

These are the final instances to which we can apply the algorithm assumed by Lemma \ref{lem:low-deg-reduc} to prove the result:

\begin{proof}[Proof of Lemma \ref{lem:low-deg-reduc}]
Setting $\bar\Delta = n^{23\zeta}+ O(n^{13.1\zeta})$, we have reached, in $O(1)$ low-space \MPC rounds, a collection of instances where all nodes $v$ have degree $d^*(v) \le \bar\Delta$ and palette size at least $\max\{d^*(v)+1, n^{23\zeta}-O( n^{13\zeta}) \} \ge \max\{\deg(v)+1, \bar\Delta-\bar\Delta^{0.6}. \}$ Furthermore, these instances are organized into $O(1)$ groups with fully-disjoint palettes (which can be solved concurrently). So, applying the assumed $O(\log\log\log n)$-round algorithm to each group in turn, we solve $(\Delta+1)$ list coloring on the original input graph in $O(\log\log\log n)$ rounds.
\end{proof}
}%

\section{Improved Coloring via Derandomization of CLP}\label{sec:derand}
In this section, we show a deterministic coloring algorithm that runs in $O(\log\log\log n)$ rounds, by derandomizing the $(\Delta+1)$-list coloring algorithm of \cite{chang2020distributed}.
%
%
%
Throughout, we assume that $\Delta \in [\log^c n, n^{\sparam/c}]$ for a sufficiently large constant $c$. Due to Lemma
\ref{lem:low-deg-reduc}, this can be assumed, \emph{almost} without loss of generality. Specifically, the guarantee of Lemma \ref{lem:low-deg-reduc} is that each $v$ has a palette of $\min\{\deg_G(v), \Delta-\Delta^{3/5}\}$ colors, instead of $\Delta+1$ colors. This was obtained also for the randomized procedures of \cite{CFGUZ19} and \cite{Parter18}. As observed in \cite{Parter18} (see Lemma \ref{lem:relaxed-CLP}), the CLP algorithm works exactly the same for that setting, up to minor modifications in the constants of the lemma statements. For simplicity of that section, we assume the standard $(\Delta+1)$ list coloring setting, and then in App. \ref{sec:small-pal} explain the minor adaptations to handle the palettes of Lemma \ref{lem:low-deg-reduc}. \\
\noindent \textbf{Smaller IDs.} Our approach is based on derandomizing $t$-round local algorithms for some constant $t$. The first preliminary step for this derandomization is to assign the nodes unique identifiers of $O(\log \Delta)$ bits, that are unique in each $t$-radius ball. To do that, the algorithm applies the well-known algorithm of Linial \cite{Linial92} to compute $\Delta^{2t}$ coloring in $G$. This can be done in $O(\log^* n)$ rounds. This part is important for the PRG based simulations as explained in Section \ref{sec:PRG}. We next iterate over the key CLP procedures and derandomize them using PRGs.

The key challenging part is in derandomizing the dense nodes in the large blocks of Section \ref{sec:no-slack}. These nodes have excess of colors due to neighbors in different classes, and thus the randomized coloring procedure is highly sensitive to the order in which the nodes get colored.


\subsection{Low-Deg Coloring}\label{sec:low-deg}
We start by providing a $O(\log\log\log n)$-round algorithm for computing $(deg+1)$-coloring in graphs with maximum degree at most $\log^c n$ for any constant $c$. This allows us later on to focus on graphs with maximum degree $\Delta \geq \log^c n$ for which we provide a derandomization of the CLP algorithm. 
\begin{lemma}\label{lem:low-deg}
For any $n$-node graph $G$ with maximum degree $d\le\log^c n$, there exists an $O(\log\log\log n)$-deterministic algorithm for computing $(\deg+1)$ list coloring. 
\end{lemma}
\begin{proof}
The algorithm derandomizes the state-of-the-art local algorithm for the $(\deg+1)$ list coloring problem of \cite{BEPS16} that runs in $T=O(\log d)+\poly(\log\log n)$ rounds (when combined with the network decomposition result of \cite{RG20}). First, the \MPC\ algorithm allocates a machine $M_v$ for every node $v$ which collects the $3T$-radius ball of $v$ in $G$. Since $d^T=n^{o(1)}$ this fits the local space of each machine. Via the standard graph exponentiation technique, these balls can be collected in $\log T$ rounds. 

Since the $T$-round randomized algorithm of \cite{BEPS16} and \cite{RG20} employs a polynomial(in $d$) time computation at each node, we can again use the PRG machinery. In what follows we provide an $O(\log\log\log n)$-round \MPC\ procedure to color $(1-1/n^{\alpha})$ of the nodes. Repeating this procedure on the remaining uncolored graph guarantees that within $O(1/\alpha)$ repetitions, all nodes are colored.

The algorithm of \cite{BEPS16} has two steps. The first is a randomized pre-shattering procedure that in $O(\log d)$ rounds guarantees that every component has size at most $\poly(d)\log n$, w.h.p. This holds even if the random decisions of nodes at distance $c$ for some constant $c$ are adversarially correlated. That is, it sufficient to have independence in each $O(1)$-radius ball. This implies that we can continue using our $O(\log \Delta)$-bit identifiers that are unique in each $c$-radius ball for the purposes of derandomization via PRG. 

The second step completes the coloring by applying the $\poly(\log N)$ deterministic algorithm for coloring for $N=\poly(\log n)$. Since we have already collected the $T$-radius ball around each node $v$ onto a machine $M_v$, each machine $M_v$ can now simulate this second step immediately. 

It remains to focus on the derandomization of the pre-shattering step. Recall that each machine $M_v$ stores the $2T$-radius ball of $v$, collected in $O(\log T)$ rounds. The machine $M_v$ computes the PRG function $\mathcal{G}$ that $\epsilon$-fools every $\poly(d)$ time computation for $\epsilon=1/(2n^{\alpha})$ for some small constant $\alpha \in (0,1)$. Specifically, the parameter $\alpha$ is chosen using Lemma \ref{lem:prg-alg} in a way that guarantees that the computation of the function $\mathcal{G}$ can be done in space of $n^{\sparam}$. This yields a seed of length $d <\sparam\log n$.  Simulating the $T$-round randomized algorithm with a random $d$-length seed guarantees that each node remains uncolored with probability of at most $1/n^{\alpha}$. The machines can then compute a seed $Z^*$ that matches this expected value, and all machines simulate the $T$-round randomized algorithm using $Z^*$. 
\end{proof}

\subsection{Initial Slack Generation}\label{sec:slack}
The CLP algorithm starts by applying a $O(1)$-round randomized procedure, called $\OneShotColoring$, which translates the sparsity in nodes' neighborhoods into a slack (excess) in the number of colors. In this procedure, each uncolored node participates with probability $p$, and each participating node $v$ randomly selects a color $c(v)$ from its palette $\Pal(v)$. This color is selected for $v$ only if none of its neighbors picked that color. 

\begin{lemma}\label{thm:one-shot-random}[Lemma 2.5 from \cite{chang2020distributed}]
	There is a $O(1)$-round randomized LOCAL algorithm that colors a subset of the nodes $V$, such that the following are true for every node $v$ with $deg(v)\geq (5/6)\Delta$:
	\begin{itemize}
		\item{(P1)} With probability $1-1/n^c$, the number of uncolored neighbors of $v$ is at least $\Delta/2$.
		\item{(P2)} With probability $1-1/n^c$, $v$ has at least $\Omega(\epsilon^2 \Delta)$ excess colors, where $\epsilon$ is the highest value such that $v$ is $\epsilon$-sparse for $\epsilon\geq 1/\Delta^{1/10}$.
	\end{itemize}
\end{lemma}

We wish to derandomize this procedure. We sketch our method here; full proofs are deferred to Appendix \ref{sec:genslack}.

We will partition nodes into $\Theta(1)$ groups using bounded-independence hash functions, in such a way that we ensure that the local sparsity property is maintained \emph{within each group}. We will then apply the PRG to directly derandomize Lemma \ref{thm:one-shot-random} within each group. The reason for the partitioning is that the error of the PRG will cause some nodes to fail to meet the criteria. However, since we have multiple groups, and any one group can provide a node with sufficient slack, we will be able to show that the PRG will succeed in at least one group for each node. For this reason our PRG derandomization here is simpler than its applications to other parts of the CLP algorithm: we need only ensure that, when coloring each group, the number of nodes which have not yet received sufficient slack decreases by an $n^{\alpha}$-factor. After coloring $1/\alpha = O(1)$ groups, we ensure that all nodes have sufficient slack, reaching Theorem \ref{thm:one-shot-almost}.

\begin{theorem}\label{thm:one-shot-almost}[Derandomization of Lemma 2.5 from \cite{chang2020distributed}]
	There is a $O(1)$-round deterministic low-space MPC algorithm that colors a subset of the nodes $V$, such that the following are true for every node $v$ with $deg(v)\geq (5/6)\Delta$:
	\begin{itemize}
		\item{(P1)} The number of uncolored neighbors of $v$ is at least $\Delta/2$.
		\item{(P2)} $v$ has at least $\Omega(\epsilon^2 \Delta)$ excess colors, where $\epsilon$ is the highest value such that $v$ is $\epsilon$-sparse for $\epsilon\geq 1/\Delta^{1/10}$.
	\end{itemize}
\end{theorem}

\subsection{Dense Coloring with Slack}\label{sec:small-med-col}
This section handles the collection of the dense nodes in the small and medium size clusters, namely, a set $S \in \{V^S_1, V^M_1, V^S_{2+}, V^M_{2+}\}$. Each of these nodes has sufficiently many neighbors not in the current node set $S$, which provides a slack in the number of available colors (regardless of how we color $S$).

The set $S$ is a collection of clusters $S_1,\ldots, S_g$ of weak diameter $2$. It is assumed that the edges within $S$ are oriented from the sparser to the denser endpoint, breaking ties by comparing IDs. For $v \in \bigcup_j S_j$, let $N_{out}(v)$ be the outgoing\footnote{An edge $e=\{u,u'\}$ is oriented as $(u,u')$ if $u$ is at layer $i$, $u'$ at layer $i'$ and $i>i'$, or if $i=i'$ and $ID(u)>ID(v)$.} neighbors of $v$ in $S$.
The dense coloring CLP procedures for these sets are based on applying procedure $\DenseColoringStep$ (version 1) which works as follows.
For each cluster $S_j$, a leader node $u \in S_j$ collects its cluster nodes and their current palettes, and colors the nodes in $S_j$ sequentially according to some random permutation. For each node $v$ (in that ordering), the leader picks a free color uniformly at random from the list of available colors of $v$. This color is set as permanent, if it does not conflict with the colors of $N_{out}(v)$.
While it is not so clear how to derandomize this procedure efficiently (i.e., in polynomial time) using the common derandomization techniques (such as bounded-independence hash functions), we show in this section that using PRGs a derandomization is possible, but at the cost of exponential local computation.

We start by describing the dynamics of the coloring procedure for dense nodes under the \emph{randomized} procedure $\DenseColoringStep$ of \cite{chang2020distributed}, and then explain how to derandomize it within $O(1)$ number of rounds. Every dense node is associated with two parameters which governs its coloring probability in the
$\DenseColoringStep$ procedure when applied simultaneously by a given set $S=S_1 \cup \ldots \cup S_g$:
\begin{itemize}
\item A parameter $Z_v$ which provides a lower bound on the number of excess colors of $v$ w.r.t the nodes in $S$. I.e., the palette size of $v$ minus $|N_{out}(v)\cap S|$ is at least $Z_v$.

\item A parameter $D_v$ which provides an upper bound on the external degree of $v$, i.e., $|N_{out}(v)\setminus S_j|\leq D_v$ where $S_j$ is the cluster of $v$.
\end{itemize}
The ratio between these bounds, denoted by $\delta_v=D_v/Z_v$, determines the probability that a node $v \in S_j$ remains uncolored after a single application of the $\DenseColoringStep$ procedure.
Two clusters $S_i$ and $S_j$ are \emph{neighbors} if there exist $u \in S_i$ and $v \in S_j$ such that $(u,v)\in E(G)$. For any positive integer $r \geq 1$, let $N^r(S_j)$ be the $r$-hop neighboring clusters of $S_j$ in $S$. For $r=1$, we may omit the index and simply write $N(S_j)$ to denote the immediate cluster neighbors of $S_j$ in $S$.

\begin{theorem}\label{thm:dense-col-slack-main-two}[Derandomization of Lemma 4.2 of \cite{chang2020distributed}]
Let $S \in \{V^S_{2+}, V^M_{2+}\}$. Suppose that each layer-$i$ node $v \in S$ has at least $\Delta/(2\log(1/\epsilon_i))$ excess colors w.r.t $S$. Then, there exists a $O(1)$-round deterministic algorithm that colors a subset of $S$ meeting the following condition. For each node $v \in V^*$, and for each $i \in [2,\ell]$, the number of uncolored $i$-layer neighbors of $v$ in $S$ is at most $\epsilon_i^5 \Delta$.
\end{theorem}
\begin{proof}
The randomized algorithm of Lemma 4.2 of \cite{chang2020distributed} is based on having $K=6$ applications of the $\DenseColoringStep$ procedure where each node $v \in S$ uses the same initial values of $Z_v,D_v,\delta_v$. After each application $k \in \{1,\ldots, 6\}$, it is shown that for every $i \in [2,\ell]$, w.h.p., every node $v$ has at most $t(k,i)=\min\{(2\delta_i)^{k-1}\Delta, \epsilon_i^5 \Delta\}$ uncolored layer-$i$ neighbors where $\delta_i=2\epsilon_i\log(1/\epsilon_i)$ (i.e., the $\delta$ bound for $i$-layer nodes in $S$).

The derandomization has $K$ phases, each phase employs $Q=O(1)$ applications of a \emph{partial} derandomization of the $\DenseColoringStep$ procedure. At the end of each phase $k \in \{1,\ldots, K\}$, the algorithm extends the coloring of the nodes in $S$ such that each node $v \in S$ has at most $t(k,i)$ uncolored layer-$i$ neighbors.
We say that a node $v$ \emph{satisfies} the $k^{th}$ invariant if it has at most $t(k,i)$ uncolored layer-$i$ neighbors where $t(1,i)=\Delta$ for every $i \in [2,\ell]$.

At the beginning of phase $k$, it is assumed that all the nodes satisfy the $k^{th}$ invariant, this clearly holds for $k=1$. The algorithm then applies $Q$ steps to partially color the nodes in $S$, such that after each step $q$, at least $1-1/n^{\alpha \cdot q}$ fraction of the nodes satisfy the $(k+1)$-invariant, where $\alpha$ is a small constant (as explained in Section \ref{sec:PRG}). A node $v \in S$ is said to be \emph{happy} if either it is colored, or else satisfies the $(k+1)$-invariant. In step $q$, the algorithm is given a subset $V_q$ of unhappy nodes, and outputs a subset $V'_q \subseteq V_q$ of happy nodes, such that $|V'_q|\geq (1-1/n^{\alpha})|V_q|$ as described next. Initially, $V_1=S$.

A single application of the \local\ randomized (and nice) $\DenseColoringStep$ procedure\footnote{That local algorithm is applied by the remaining \emph{uncolored} part of each cluster.} on the remaining uncolored nodes, using the same initial bounds of $Z_v,D_v,\delta_v$, makes a node $v \in V_q$ happy with high probability. Note that we can indeed use the same $Z_v,D_v,\delta_v$ bounds as the external uncolored degree is non-increasing and the number of excess colors is non-decreasing. Therefore, using the PRG machinery, one can then simulate the $\DenseColoringStep$ procedure with a random seed of length $d<\sparam\log n$ up to an additional additive error of $1/(2n^{\alpha})$ for some small constant $\alpha$, which makes at least $(1-1/n^{\alpha})$ fraction of the nodes in $V_q$ happy. The machines can then compute a seed $Z^* \in \{0,1\}^d$ that match this expected value. This is done as follows. We allocate a machine $M_j$ for every cluster $S_j$, that machine also stores the $2$-hop neighboring clusters of $S_j$ in $S$, denoted by $N^2(S_j)$. The machine $M_j$ simulates the $\DenseColoringStep$ procedure on the currently uncolored nodes in $S_j \cup N^1(S_j)$ with each possible seed $Z \in \{0,1\}^d$, and evaluates the number of happy nodes in $V_q \cap S_j$ under each such seed.
Letting $Z^*$ be the seed that maximizes the number of happy nodes, then all machines $M_j$ adopt the colors obtained by simulating the $\DenseColoringStep$ procedure on $S_j$ with $Z^*$. Denote by $V'_q \subseteq V_q$ the set of happy nodes under $Z^*$, then the output of the phase is the set $V_{q+1}=V'_q \setminus V_q$ consisting of the remaining unhappy nodes. This completes the description of step $q$.

Since for every step $q \in \{1,\ldots, Q\}$, at least $(1-1/n^{\alpha})$ fraction of the nodes become happy, after $Q=O(1/\alpha)$ steps, all nodes are happy and satisfy the $(k+1)^{th}$ invariant. Therefore after $K=6$ phases, each node has at most $t(K,i)\leq \epsilon_i^5\Delta$ uncolored layer-$i$ neighbors, for every $i \in [2,\ell]$, the claim follows.
\end{proof}

To handle the layer-1 nodes in small and medium clusters, we derandomize Lemma 4.3 of \cite{chang2020distributed}. The randomized procedure of that lemma colors \emph{all} the nodes in these clusters with high probability. We next show that we can color all these nodes deterministically. The argument applies the PRG machinery in a similar manner to Theorem \ref{thm:dense-col-slack-main-two}. Missing proofs are deferred to Appendix \ref{sec:missing-derand}.
\begin{theorem}\label{thm:dense-col-slack-main-one}[Derandomization of Lemma 4.3 of \cite{chang2020distributed}]
Let $S \in \{V^S_{1}, V^M_{1}\}$. Suppose that each node $v \in S$ has at least $\Delta/(2\log(1/\epsilon_i))$ excess colors w.r.t $S$. There exists a $O(1)$-round deterministic algorithm that colors all the nodes of $S$.
\end{theorem}
\begin{proof}
The algorithm has two phases. In the first phase, it derandomizes a single application of the $\DenseColoringStep$ procedure. This will reduce the uncolored degree of each node in $S$ to $O(\Delta^{9/10})$. Since every node in $S$ has $\Omega(\Delta/\log \Delta)$ excess colors, the second phase exploits this large gap between the excess colors and the current degree, and employs (a derandomization of) the sparse-coloring algorithm to complete the coloring within $O(1)$ number rounds.
We start with the first phase. Since $\epsilon_1=\Delta^{-1/10}$, for each $v \in S$ it holds that
$$Z_v=\Delta/(2\log(1/\epsilon_1))~, D_v=\epsilon_1 \cdot \Delta, \mbox{~and~} \delta_v=2\epsilon_1\log(1/\epsilon_1)~.$$

By the proof of Lemma 4.3 in \cite{chang2020distributed}, after a single application of the $\DenseColoringStep$ procedure, w.h.p., the uncolored degree of each node in $S$ is at most $O(\Delta^{9/10}\log \Delta)$.
Our derandomization is based on the observation that the $\delta_v$ bound of each node is non-increasing. To see this observe that the number of excess colors is non-decreasing regardless of how we extend the given coloring. In addition, the (uncolored) external degree cannot increase as well, and consequently, the ratio between these two terms $\delta_v$ is non-increasing. This would be important for our arguments, as in our context, we derandomize $O(1)$ applications of a weak randomized simulation of the $\DenseColoringStep$ procedure in which each node satisfies the desired degree bound with probability of $1-1/n^{\alpha}$ for a small constant $\alpha$. We now explain the algorithm in more details. Throughout, a node $v$ is \emph{happy} if either it is colored, or else its uncolored degree is at most $O(\Delta^{9/10}\log \Delta)$.

The algorithm has $Q=O(1)$ phases where in each phase we deterministically extend the coloring to the nodes in $S$, making a $(1-1/n^{\alpha})$-fraction, of the yet uncolored nodes, happy. Let $V_q$ be the current set of unhappy nodes at the beginning of the $q^{th}$ phase, for $q \in \{1,\ldots, Q\}$, where initially $V_1=S$. Since the bounds on $Z_v,D_v,\delta_v$ still hold for every node $v \in V_q$, a single application of the $\DenseColoringStep$ procedure makes each node $v \in V_q$ happy w.h.p.

Since the local computation of the $\DenseColoringStep$ procedure is polynomial, using the PRG machinery, one can simulate the $\DenseColoringStep$ procedure with a seed of length $d<\sparam\log n$ up to an additive error of $1/(2n^{\alpha})$. Therefore, simulating the $\DenseColoringStep$ procedure with random seed of length $d$ provides at least $(1-1/n^{\alpha})|V_q|$ happy nodes. We next explain how to compute in $O(1)$ rounds a seed $Z^* \in \{0,1\}^d$ that match this expected value of happy nodes.

For every cluster $S_j \subseteq S$, we allocate a machine $M_j$ which stores the cluster $S_j$ and its \emph{two-hop} cluster neighbors $N^2(S_j)$ in $S$. This allows $M_j$ to simulate the $\DenseColoringStep$ procedure on the clusters $S_j \cup N^1(S_j)$ under each possible seed $Z \in \{0,1\}^d$ and to evaluate the number of happy nodes in $S_j \cap V_q$ under each such seed. We then allocate a machine $M_Z$ for every seed $Z \in \{0,1\}^d$, and pick the seed that achieves the largest number of happy nodes. This can be computed in $O(1)$ number of rounds by sorting.
Letting $Z^*$ be the elected seed, all machines of all uncolored nodes simulate the $\DenseColoringStep$ procedure with the seed $Z^*$. Let $V'_q \subseteq V_q$ be the happy nodes under $Z^*$. Since $Z^*$ achieves at least the expected number of happy nodes, we have that $|V'_q|\geq (1-1/n^{\alpha})|V_q|$. The input to the next phase $q+1$ is given by $V_{q+1}=V'_q \setminus V_q$. By having $Q=O(1/\alpha)$ phases, we have that $V_Q=\emptyset$ and thus all nodes are happy.  At this point, the maximum degree in the remaining uncolored subgraph is $O(\Delta^{9/10}\log \Delta)$.
The second step of the algorithm employs the algorithm of Cor. \ref{cor:sparse-gap-col} which deterministically colors the remaining nodes in $O(1)$ number of rounds. This completes the proof.
\end{proof}


\subsection{Dense Coloring without Slack}\label{sec:no-slack}
In this section we consider the more challenging subset of dense nodes, in the large blocks $V^L_1$ and $V^{L}_{2+}$, that have no slack in their colors due to neighbors in other subsets. The randomized CLP procedure colors these nodes by employing a modified variant of the dense coloring procedure, denoted as $\DenseColoringStep$ (version 2).

Our analysis for coloring $V^L_1$ and $V^{L}_{2+}$ will be very similar but have slight differences. So, we first present a framework that applies to both cases, and then show how it is instantiated to color $V^{L}_{2+}$ and $V^L_1$ in Sections \ref{sec:noslack1} and \ref{sec:large-col-one} respectively.

Let $S \in \{V^L_1, V^L_{2+}\}$. For every node $v \in S$, let $N^*(v)$ be the neighbors of $v$ in $S$ of layer number smaller than or equal to the layer number of $v$. Let the clusters of $S=S_1 \cup \ldots \cup S_g$ be ordered based on nondecreasing order their layer number. Each cluster, and each node in a cluster have an ID that is consistent with this ordering. We use the term \emph{antidegree} of $v \in S_j$ to the number of \emph{uncolored} nodes in $S_j \setminus (N(v)\cup \{v\})$. The term \emph{external degree} of $v \in S_j$ refers to the number of uncolored nodes in $N^*(v) \setminus S_j$. The rate by which a dense cluster $S_j \in S$ gets colored is determined by the following parameters:
\begin{itemize}
\item A parameter $D_j$ that provides an upper bound on the uncolored external degree and antidegree\footnote{In contrast to Section \ref{sec:small-med-col}, here the bounds depend on $N^*(v)$ rather than on $N_{out}(v)$.} of each node $v \in S_j$. That is, $|N^*(v)\setminus S_j|\leq D_j$ and $|S_j \setminus (N^*(v) \cap \{v\})|\leq D_j$, respectively.

\item A lower bound $L_j$ on the size of the uncolored part of $S_j$.

\item An upper bound $U_j$ on the size of the uncolored part of $S_j$.

\item A shrinking rate $\delta_j\geq D_j \log(|U_j|/D_j)/|L_j|$ that determines the speed at which the cluster $S_j$ shrinks (due to the coloring of its nodes).

\item An upper bound $t(k,i)$ on the number of uncolored layer-$i$ neighbors, for every $i \in [2,\ell]$ (in iteration $k$).
\end{itemize}

In the randomized \LOCAL algorithm of CLP, to color blocks of this type, version 2 of the $\DenseColoringStep$ procedure is again run by each cluster leader and works as follows. First, the leader of cluster $S_j$ picks a $1-\delta_j$ fraction of the nodes in $S_j$ uniformly at random, and computes a permutation on the elected nodes. It then iterates over these nodes according to the permutation order, picking a free color uniformly at random for each such node. These colors are fixed as the permanent colors, only if there are no collisions with their outgoing external neighbors. The dense coloring CLP procedure is based on having multiple applications of this $\DenseColoringStep$ procedure. The key property that underlies the correctness of the $\DenseColoringStep$ procedure is summarized in the next lemma.

\begin{lemma}\label{lem:dense-step-var-two}[Lemma 6.2 of \cite{chang2020distributed}]
Consider an execution of $\DenseColoringStep$ (version 2). Let $T$ be any subset of $S$ and let $\delta=\max_{j: S_j \cap T \neq \emptyset}\delta_j$. For any number $t$, the probability that the number of uncolored nodes in $T$ is at least $t$ is at most ${|T| \choose t}\cdot (O(\delta))^t$.
\end{lemma}

The CLP procedure makes $O(1)$ iterations of $\DenseColoringStep$ (version 2), and its analysis defines a sequence of invariants which are shown to be satisfied after each of these iterations. We use a very similar set of invariants, which we now describe. 

\paragraph{The invariants.} We will also have $O(1)$ iterations of applying the $\DenseColoringStep$ procedure. At the beginning of each iteration $k$, each cluster is required to satisfy certain desired properties concerning the uncolored external and antidegrees of its nodes, and the size of the uncolored cluster.

Specifically, at the beginning of iteration $k$, each cluster $S_j$ is required to satisfy some properties concerning its \emph{own} coloring; we will call this the \textbf{self-invariant}. The self-invariant will include the following criteria:

\begin{itemize}
\item  the antidegree of nodes in $S_j$ is at most $D_j^{(k)}$;
\item  the number of uncolored nodes in $S_j$ is in the range $[L_j^{(k)}, U_j^{(k)}]$.
\end{itemize}

We will also define a \textbf{neighbor-invariant}, containing the following properties (concerning the coloring of neighboring clusters):

\begin{itemize}
\item nodes in $S_j$ have at most $D_j^{(k)}$ uncolored neighbors in $N^*(v) \setminus S_j$ (i.e., bound on external degree),

\item nodes in $S_j$ have at most $t(k,i)$ uncolored layer-$i$ neighbors, for every $i \in [2,\ell]$.
\end{itemize}

\hide{
We classify the invariants of the CLP algorithm into two types: \emph{self-invariant} and \emph{neighbor-invariant}.  Roughly speaking, the self-invariant property of a cluster $S_j$ depends only on the coloring status of the nodes of the cluster (e.g., all nodes in $S_j$ should have an antidegree at most $x$). In contrast, the neighbor-invariant of $S$ depends only on the coloring status of the neighboring clusters $N(S_j)$ of $S_j$ (e.g., all nodes in $S_j$ should have an external degree at most $x$), where $N(S_j)$ are all the clusters that have at least one neighbor of $S_j$ in the current collection of clusters considered. The algorithm defines $K=O(1)$ invariants where the $k^{th}$ phase of the algorithm assumes that all nodes satisfy the $k^{th}$ invariant. The $k^{th}$ invariant provides lower bound value $L^{(k)}_j$ and an upper bound $U^{(k)}_j$ on the current uncolored size of the cluster $S_j$. In addition, it provides an upper bound on the (i) external uncolored degree of a node in a cluster, (ii) uncolored antidegree of a node in a cluster, and possibly also on (iii) number of layer-$i$ uncolored neighbors of each node in a cluster.}

If a cluster satisfies both the $k^{th}$ self-invariant and the $k^{th}$ neighbor-invariant, we say it satisfies the $k^{th}$ \textbf{invariant}. While the general form of our invariants is the same as in CLP (and they have only minor differences in the values $D_j^{(k)}$, $L_j^{(k)}$, $U_j^{(k)}$ etc.), the distinction between self- and neighbor-invariants does not appear in CLP and is introduced here since we must treat these properties differently during our derandomization.

For neighbor-invariants we also introduce the concept of $\gamma$-satisfying the invariant: that is, the upper bounds comprising the invariant are multiplied by a factor of $\gamma$. The purpose of this is to introduce extra tolerance in the invariants: we will show that each of the $O(1)$ iterations only increases $\gamma$ by a constant factor, and therefore the final invariants satisfied are still sufficient for the CLP analysis to proceed.

\begin{definition}
A cluster $S_j$ $\gamma$-satisfies the $k^{th}$ neighbor-invariant if it satisfies the neighbor-invariant up to a multiplicative factor of $\gamma$, i.e. nodes in $S_j$ have at most $\gamma D_j^{(k)}$ uncolored neighbors not in $S_j$, and at most $\gamma t(k,i)$ uncolored layer-$i$ neighbors, for every $i \in [2,\ell]$.
\end{definition}

In the \emph{randomized} CLP procedures of coloring $V^L_1$ and $V^L_{2+}$, the invariants hold with high probability (since $\Delta\geq \log^c n$). Specifically, for $\Delta> \log^c n$, the CLP analysis shows that after the $k^{th}$ application of the procedure $\DenseColoringStep$ all blocks satisfy the $(k+1)^{th}$ invariant with probability of $1-1/n^c$. Due to the sublinear space limitation of our machines, when using a PRG for randomness, we will satisfy these properties with probability of only $1-1/n^{\alpha}$ (using a random seed of length at most $\sparam \log n$), for some (small) constant $\alpha$. Handling this increased error is the main difficulty of our derandomization.

To do so, we set out a framework for coloring $V^L_{1}$ and $V^L_{2+}$ describing properties which we will show are satisfied by the $\DenseColoringStep$ procedure (though for different values used within the bounds in the respective invariants).

\begin{framework}[Dense Coloring with Slack]\label{frame:dcs}
Let $\mathcal{S}' \subseteq \mathcal{S}$ be a subset of clusters
such that all clusters in $\mathcal{S}\setminus \mathcal{S}'$ $r$-satisfies the $(k+1)^{th}$ invariant, and
in addition, each cluster $S_j \in \mathcal{S}'$ satisfies:
\begin{itemize}
	\item{(P1)} the $(k+1)^{th}$ \emph{neighbor-invariant}, up to a multiplicative factor of $r$, w.r.t its neighbors in $\mathcal{S} \setminus \mathcal{S}'$ (if such exist).
	\item{(P2)} the $k^{th}$ \emph{neighbor-invariant} w.r.t its neighbors in $\mathcal{S}'$, and
	\item{(P3)} the $k^{th}$ \emph{self-invariant}.
\end{itemize}
Applying procedure $\DenseColoringStep$ \emph{only} to the nodes in $\mathcal{S}'$, we will show that w.h.p:
\begin{itemize}
	\item{(P4)} every $S_j \in \mathcal{S}'$  $(r+1)$-satisfies the $(k+1)^{th}$ invariant; and
	\item{(P5)} every neighboring cluster $S_{j'} \in \bigcup_{S_j \in \mathcal{S}'} N(S_j)$ (where $S_{j'}$ is
	possibly in $\mathcal{S}\setminus \mathcal{S}'$) $(r+1)$-satisfies the $(k+1)^{th}$ (neighbor) invariant.
\end{itemize}

\end{framework}

We then give the following lemma demonstrating how to derandomize the procedure using the PRG, in such a way that only a small (though larger than under full randomness) subset of clusters do not satisfy the next invariant.

\begin{lemma}\label{det-dense-step}
Let $\mathcal{S}' \subseteq \mathcal{S}$ be a collection of clusters that satisfy the properties of Framework \ref{frame:dcs}. Then,\hide{ given the procedure $\DenseColoringStep$ \ref{frame:dcs},} one can provide a $O(1)$-round deterministic \MPC\ procedure that extends the coloring (of nodes in $\mathcal{S}'$), resulting in a subset $\mathcal{S}''\subseteq \mathcal{S}'$ such that (i) $|\mathcal{S}''|\geq (1-1/n^{\alpha})|\mathcal{S}'|$ and in addition (ii) the clusters of $\mathcal{S}''$ satisfy (P4,P5).
%
\end{lemma}
\begin{proof}
The deterministic algorithm is applied only by the machines that store the nodes in the clusters of $\mathcal{S}'$. The remaining machines are idle.
Since the clusters in $\mathcal{S}'$ already $r$-satisfy the $(k+1)^{th}$ neighbor-invariant w.r.t their neighbors in $\mathcal{S}\setminus \mathcal{S}'$, there is no need to further color the nodes in the $\mathcal{S}\setminus \mathcal{S}'$ clusters.

Our goal is to extend the coloring of the nodes in the $\mathcal{S}'$ clusters in order to demonstrate a large collection of clusters $\mathcal{S}''\subseteq \mathcal{S}'$ that $(r+1)$-satisfies the $(k+1)^{th}$ invariant, and also their cluster neighbors (possibly in $\mathcal{S}\setminus \mathcal{S}'$) $(r+1)$-satisfy the $(k+1)^{th}$ invariant. From now on, we call a cluster in $\mathcal{S}'$ \emph{happy} if the cluster and all its neighboring clusters $(r+1)$-satisfy the $(k+1)^{th}$ invariant.  By Framework \label{frame:dcs}, we get that w.h.p., every cluster in $\mathcal{S}'$ is happy.

In our derandomization scheme, using the PRG, we first provide a weaker randomized algorithm that uses a random seed of length $d<\sparam\log n$ and errs with probability of $\epsilon=1/(2n^{\alpha})$ for some constant $\alpha \in (0,1)$. Using this pseudorandom procedure, each cluster becomes happy with probability of at least $p=1-1/n^c-\epsilon\geq 1-1/n^{\alpha}$.
The goal is then to compute a $d$-length seed that attains (at least) the expected value of $(1-1/n^{\alpha})\cdot |\mathcal{S}'|$ of happy clusters.

The algorithm allocates a machine $M_j$ for every cluster $S_j \in \mathcal{S}'$ that stores the cluster $S_j$ and the $3$-hop neighboring clusters of $S_j$  in $\mathcal{S}$. The machine of $M_j$ evaluates the happiness of $S_j$ under every seed $Z \in \{0,1\}^d$. This is done by first computing the PRG function $\mathcal{G}^*$ that $(\poly(\Delta),1/(2n^{\alpha}))$-fools the family of all $\poly(\Delta)$-time algorithms. Then for every $Z \in \{0,1\}^d$ it simulates the randomized $\DenseColoringStep$ procedure on the $3$-hop neighboring clusters of $S_j$ using $\mathcal{G}(Z)$ as the source of random coins. This allows $M_j$ to determine the coloring of all nodes in $S_j$ and in the $2$-hop neighboring clusters of $S_j$, and thus to determine the happiness of $S_j$ and its $1$-hop neighboring clusters.
The machines then elects the seed $Z^*$ that maximizes the number of happy clusters, and fixes the coloring to be that generated by applying $\DenseColoringStep$ using $\mathcal{G}^*(Z^*)$.  The number of happy clusters is at least $(1-1/n^{\alpha})$ fraction of the clusters in $\mathcal{S}'$ as desired.
\end{proof}

Finally, we can show that repeated (but still $O(1)$) applications of the derandomized $\DenseColoringStep$ procedure can successfully color all clusters (at the cost of a $\lceil 1/\alpha \rceil=O(1)$ factor in the invariant bounds).

\begin{theorem}\label{thm:det-one-iteration}
Assuming that the given collection of clusters $\mathcal{S}$ satisfies the $k^{th}$ invariant, there exists a $O(1)$-round deterministic algorithm that extends the current coloring such that all  $\mathcal{S}$ clusters $(\lceil 1/\alpha \rceil)$-satisfy the $(k+1)^{th}$ invariant.
\end{theorem}
\begin{proof}
We make $Q=\lceil 1/\alpha \rceil$ applications of the algorithm of Lemma \ref{det-dense-step} as follows. For every $q \in \{1,\ldots, Q\}$, in the $q^{th}$ application of Lemma \ref{det-dense-step}, let $\mathcal{S}'=\mathcal{S}_q$ where initially $\mathcal{S}_1=\mathcal{S}$, and denote the output subset by $\mathcal{S}'_q \subseteq \mathcal{S}_q$ (referred to as $\mathcal{S}''$ in Lemma \ref{det-dense-step}). The algorithm then maintains the colors assigned in that application to the nodes in the clusters of $\mathcal{S}'_q$, but \emph{omit} all the colors computed in that application to the nodes in the clusters of $\mathcal{S}_q \setminus \mathcal{S}'_q$. The subset $\mathcal{S}_{q+1}=\mathcal{S}'_q \setminus \mathcal{S}_q$ is provided as the input the $(q+1)^{th}$ application. This completes the description of the algorithm, we turn to analyze its correctness.

We first prove by induction on $q$ that (i) the input to the $q^{th}$ application of Lemma \ref{det-dense-step}, namely $\mathcal{S}_q$, satisfies the properties (P1-P3) required by the lemma with $r=q-1$, and that (ii) at the end of the $q^{th}$ application (before omitting the colors assigned to nodes in $\mathcal{S}_q \setminus \mathcal{S}'_q$ in that application), all the clusters in $\bigcup_{a=1}^{q} \mathcal{S}'_a$ and their neighboring clusters $q$-satisfy the $(k+1)$-invariant.

Initially, $\mathcal{S}'=\mathcal{S}_1=\mathcal{S}$, and the properties hold trivially, as all clusters in $\mathcal{S}$ satisfy the $k^{th}$ invariant. By the output of Lemma \ref{det-dense-step}, the clusters in $\mathcal{S}'_1$ and their neighbors satisfy the $(k+1)^{th}$ invariant. Assume that both claims hold for every $a \in \{1,\ldots, q\}$ and consider the $(q+1)^{th}$ application.

(i) By the induction hypothesis for $q$, the collection of clusters $\bigcup_{a=1}^q \mathcal{S}'_{a}$ and their neighbors $q$-satisfy the $(k+1)^{th}$ invariant. Then, before applying the $(q+1)^{th}$ application, the algorithm omits the colors assigned in the $q^{th}$ application to nodes in $\mathcal{S}_{q} \setminus\mathcal{S}'_{q}$. Since the colors of nodes in $\bigcup_{a=1}^q \mathcal{S}'_{a}$ are kept, every cluster in $\mathcal{S}_{q+1}$ still $q$-satisfies the $(k+1)^{th}$ neighbor-invariant w.r.t its neighbors in
$\mathcal{S}\setminus \mathcal{S}_{q+1}=\bigcup_{a=1}^q \mathcal{S}'_{a}$. This holds as this invariant only depends on the colors of nodes in $\mathcal{S}\setminus \mathcal{S}_{q+1}$. Thus (P1) holds. In addition, since the colors of all nodes in $\mathcal{S}_{q+1}$ are the same as in the beginning of the algorithm, these clusters still maintain the $k^{th}$ invariant so (P2,P3) hold.

(ii) By the output of the $(q+1)^{th}$ application, the clusters in $\mathcal{S}'_{q+1}$ and their neighboring clusters $(q+1)$-satisfy the $(k+1)^{th}$ invariant. Note that the only clusters in $\mathcal{S}\setminus \mathcal{S}_{q+1}$ that have neighbors in $\mathcal{S}_{q+1}$ are the clusters in $\mathcal{S}'_{q+1}$.
Thus by combining with the induction assumption, we have that all clusters in $\mathcal{S}\setminus \mathcal{S}_{q+1}$ and their neighbors $(q+1)$-satisfy the $(k+1)^{th}$ invariant. The induction step holds.

Next, we show that after $Q$ applications, $\mathcal{S}_Q=\emptyset$, and thus all clusters in $\mathcal{S}$ $Q$-satisfy the $(k+1)^{th}$ invariant. We prove by induction on $q$ that $|\mathcal{S}_q|\leq 1/n^{(q-1)\alpha}\cdot |\mathcal{S}|$. The base of the induction holds as $|\mathcal{S}_1|\leq |\mathcal{S}|$. Assume that it holds up to $a \leq q$. The output of the $q^{th}$ application satisfies that $|\mathcal{S}'_q|\geq (1-1/n^{\alpha})\cdot |\mathcal{S}_q|$. Since $\mathcal{S}_{q+1}=\mathcal{S}'_q \setminus \mathcal{S}_q$, we have that $|\mathcal{S}_{q+1}|\leq |\mathcal{S}_q|/n^{\alpha}$. By plugging in the induction assumption for $\mathcal{S}_q$, we get that
$|\mathcal{S}_{q+1}|\leq |\mathcal{S}_q|/n^{q \cdot\alpha}$ as desired. Finally, since the algorithm makes $O(1)$ applications to the algorithm of Lemma \ref{det-dense-step}, the over all round complexity is bounded by $O(1)$. The lemma follows.
\end{proof}

\subsubsection{Coloring Large Blocks of Layers $i\geq 2$}\label{sec:noslack1}
We next show how to use the general algorithm of Theorem \ref{thm:det-one-iteration} in order to color the large blocks in layers $[2,\ell]$. Let $S=S_1 \cup \ldots \cup S_g$ be the union of clusters of $V^L_{2+}$.

We must set the quantities within the bounds of the $k^{th}$ self- and neighbor-invariant for every $k \in \{1,\ldots, O(1)\}$. The ultimate goal of these invariants is to gradually reduce the number of uncolored $i$-layer neighbors for each node to at most $\epsilon_i^5 \Delta$ for every $i \in [2,\ell]$.

For every\footnote{Since the large blocks corresponds to distinct almost-cliques, all the nodes in a cluster $S_j$ belong to the same layer.} layer-$i$ cluster $S_j$, we define the bounds on the degrees and the cluster size as follows. Set $D^{(1)}_j=3\epsilon_i \Delta$, $U_j^{(1)}=(1+3\epsilon_i)\Delta$ and $L_j^{(1)}=\Delta/(\log(1/\epsilon_i)$. Note that initially it holds that $|S_j|\in [U_j^{(1)}, L_j^{(1)}]$ and the bound on the external and antidegree is at most $D^{(1)}_j=3\epsilon_i \Delta$. Also, define $\delta_j^{(1)}=\frac{D_j^{(1)}\log(L_j^{(1)}/D_j^{(1)})}{L_j^{(1)}}$. For every $i \in [2,\ell]$, let $t(1,i)=\Delta$ be the initial bound on the number of uncolored layer-$i$ neighbors of every node.
Let $\beta>0$ be a sufficiently large constant, then for every integer $q\geq 1$, define:
\begin{equation}\label{eq:bounds-even}
D^{(2q)}_j=\beta \delta_j^{(2q-1)} \cdot D_j^{(2q-1)}, ~~L^{(2q)}_j=\delta_j^{(2q-1)} \cdot L_j^{(2q-1)} \mbox{~and~} U^{(2q)}_j=\beta\delta_j^{(2q-1)} \cdot U_j^{(2q-1)}~.
\end{equation}
Let $\delta^{(2q-1)}_{*,i}$ be the minimum $\delta_j^{(2q-1)}$ value over all layer-$i$ clusters $S_j$. Then define the upper bound on the number of uncolored layer-$i$ neighbors:
\begin{equation}\label{eq:tqeven}
t(2q,i)=\max\{\delta^{(2q-1)}_{*,i} \cdot t(2q-1,i), \epsilon^6_i \cdot \Delta\}~.
\end{equation}
In addition, letting $\gamma=(2/\alpha)$, then for every integer $q \geq 1$, define the bounds:
\begin{equation}\label{eq:bounds-odd}
D^{(2q+1)}_j=\gamma \cdot D^{(2q)}_j, ~~L^{(2q+1)}_j=L^{(2q)}_j \mbox{~and~} U^{(2q+1)}_j=U^{(2q)}_j~,
\end{equation}
and for every $i \in [2,\ell]$, define $t(2q+1,i)=\gamma\cdot t(2q,i)~.$

So, the resulting invariants for coloring $V^L_{2+}$ are as follows:

\begin{definition}[The $k^{th}$ invariant]
A cluster $S_j$ satisfies the $k^{th}$ \textbf{self-invariant} if:
\begin{itemize}
	\item $|S_j| \in [L_j^{(k)},U_j^{(k)}]$, and
	\item each node in $S_j$ has (uncolored) antidegree at most $D_j^{(k)}$.
\end{itemize}

In addition, it satisfies the $k^{th}$ \textbf{neighbor-invariant} if each node in $S_j$ has at most:
\begin{itemize}
\item $D_j^{(k)}$ uncolored neighbors in $N^*(v) \setminus S_j$ (i.e., bound on external degree),

\item $t(k,i)$ uncolored layer-$i$ neighbors, for every $i \in [2,\ell]$.
\end{itemize}

A cluster $S_j$ is said to satisfy the $k^{th}$ \textbf{invariant} if it satisfies both the self- and neighbor-invariants.
\end{definition}

We see that this definition fits into the general form described in Section \ref{sec:no-slack}.

\begin{observation}\label{obs:approx-sat}
For every $q\geq 1$, every cluster $S_j$ that $\gamma$-satisfies the $(2q)^{th}$-invariant also satisfies the $(2q+1)^{th}$ invariant.
\end{observation}

We proceed by showing that the $\DenseColoringStep$ procedure (version 2) satisfies the properties of Framework \ref{frame:dcs} under these invariant definitions.
\begin{lemma}\label{lem:rand-two}
Let $k=2q-1$ for some integer $q \in [1,5]$. Then, applying $\DenseColoringStep$ procedure (version 2) to the nodes in $\mathcal{S}'$
there exists a randomized \local\ algorithm for coloring $V^L_{2+}$ that satisfies the properties of Framework \ref{frame:dcs}. Moreover, the algorithm has only two \local\ rounds of communication in the cluster-graph\footnote{The graph obtained by contracting each cluster into a super-node.}.
\end{lemma}

The proof can be found in Appendix \ref{sec:missing-derand}.
\def\APPENDRANDTWO{
\begin{proof}[Proof of Lemma \ref{lem:rand-two}]
The randomized algorithm is given by a single application of the $\DenseColoringStep$ procedure to the nodes in $\mathcal{S}'$. The main difference to the CLP analysis is that we apply the \local\ procedure only on the nodes in the clusters of $\mathcal{S}'$. Indeed this algorithm runs in $2$ rounds over the cluster graph. At the high level, since the bounds $U^{(k)}_{j}, L^{(k)}_{j}, D^{(k)}_{j}$ are at least poly-logarithmic for every cluster $S_j$ and every $k \in [1,O(1)]$, the invariants hold w.h.p. using the same argument as in the proof of Lemma 4.4 in \cite{chang2020distributed}.

To be more specific consider first a cluster $S_j \in \mathcal{S}\setminus \mathcal{S}'$. Recall that the cluster $S_j$ $r$-satisfies the $(k+1)=2q$ invariant w.r.t its neighbors in $\mathcal{S}\setminus \mathcal{S}'$. We start by bounding the number of uncolored layer-$i$ neighbors in $\mathcal{S}'$ of each node $v \in S_j$.
Fix $i \in [2,\ell]$ and a node $v \in S_j$.
Let $x_i(v)$ be the number of uncolored-$i$ neighbors of $v$ in $\mathcal{S}'$, and let $d_i(v)$ be the number of external neighbors of $v$ in $\mathcal{S}'$, before applying the $\DenseColoringStep$ procedure. If $x_i(v) \leq \log^3 n$, then since $\epsilon^5_i \Delta \geq \log^3 n$ the claim already holds. Otherwise, by the proof of Lemma 4.4 in \cite{chang2020distributed}, w.h.p., the number of uncolored-$i$ neighbors of $v$ after applying the procedure on $\mathcal{S}'$ is reduced to $t(2q,i)$ (see Eq. \ref{eq:tqeven}). Altogether, we get that w.h.p. the number of uncolored-$i$ neighbors of each node $v \in S_j$ in the nodes of $\mathcal{S}'$ is at most $t(2q,i)$. Since by the assumption, $v$ has at most $r \cdot t(2q,i)$ uncolored layer-$i$ neighbors in $\mathcal{S}\setminus\mathcal{S}'$, it has in total $(r+1)t(2q,i)$ uncolored layer-$i$ neighbors w.h.p. The bound on the external degree in the clusters of $\mathcal{S}'$ is obtained in an analogous manner. Let $d(v)$ be the number of $v$'s uncolored neighbors in the clusters of $\mathcal{S}'$ before applying the $\DenseColoringStep$ procedure.
In the case where $d(v)\leq \log^3 n$, the claim holds immediately. Assume then that $d(v)\geq  \log^3 n$. Then by
by applying Lemma \ref{lem:dense-step-var-two} taking $T$ to be the $d(v)$ neighbors of $v$ in  $\mathcal{S}'$, we get that w.h.p. the number of remaining uncolored neighbors in $\mathcal{S}'$ after applying procedure $\DenseColoringStep$ is at most $\beta d(v) \cdot \delta^{(2q-1)}_{*,i} \leq D^{(2q)}_j$, for some large constant $\beta>0$. This holds since all the external uncolored neighbors in $N^*(v)$ are in layer-$a$ for $a\leq i$, and as $\delta^{(2q-1)}_{*,i}\geq \delta^{(2q-1)}_{*,a}$ for every $a\leq i$. Combining with the neighbors of $v$ in the clusters of $\mathcal{S}\setminus \mathcal{S}'$, we get that $v$ has at most $(r+1)\cdot D^{(2q)}_j$ external uncolored neighbors. This completes the requirements for all the clusters in $\mathcal{S}\setminus \mathcal{S}'$.

Next, consider a cluster $S_j \in \mathcal{S}'$. We first consider the neighbor-invariant.
By the assumption, each $v \in S_j$ has at most $r \cdot t(2q,i)$ uncolored layer-$i$ neighbors and at most $r \cdot D^{(2q)}_j$ external neighbors in the clusters of $\mathcal{S} \setminus\mathcal{S}'$. It is therefore sufficient to show that after applying the $\DenseColoringStep$ procedure on the nodes of $\mathcal{S}'$ w.h.p. $v$ has at most $t(2q,i)$ uncolored layer-$i$ neighbors and at most $D^{(2q)}_j$ external neighbors in the clusters of $\mathcal{S}'$. The argument works by applying Lemma \ref{lem:dense-step-var-two} taking $T$ to be the uncolored layer-$i$ neighbors of $v$, and then by taking $T$ to be the external degree of $v$ in $\mathcal{S}'$ (respectively). We get that w.h.p. $S_j$ $(r+1)$ satisfies the $(2q)^{th}$ neighbor-invariant.

Finally, it remains to show that $S_j$ satisfies the $(2q)^{th}$ self-invariant. For the bound on the antidegree, consider $v \in S_j$ and apply Lemma \ref{lem:dense-step-var-two} with $T$ being the anti-neighbors of $v$. The bounds on the upper and the lower bound on the uncolored part of $S_j$ follows exactly as in \cite{chang2020distributed} (as this depends only on the randomness of the nodes in $S_j$).
\end{proof}
}
\begin{theorem}\label{thm:col-large-two}
There exists a $O(1)$-round deterministic algorithm that partially colors the nodes in $V^L_{2+}$ such that each node has at most $\epsilon_i^5 \cdot \Delta$ uncolored layer-$i$ neighbors for every $i \in [2,\ell]$.
\end{theorem}
\begin{proof}
The proof follows simply by having $Q=7$ applications of the algorithm of Theorem \ref{thm:det-one-iteration} which is based on the existence of the randomized algorithm of Lemma \ref{lem:rand-two} for coloring $V^L_{2+}$.

We prove by the induction on $q \in \{1,\ldots, Q\}$ that at the beginning of the $q^{th}$ application all nodes
clusters $\mathcal{S}$ satisfy the $(2q-1)^{th}$ invariant. The base case holds trivially as all clusters in $\mathcal{S}$ are large blocks. Assume that it holds up to $q\geq 1$ and consider the $(q+1)^{th}$ application.
By the induction assumption, at the beginning of the $q^{th}$ application,  all clusters satisfy the $(2q-1)^{th}$ invariant. By applying the algorithm of Theorem \ref{thm:det-one-iteration}, all clusters $\gamma$-satisfy the $(2q)^{th}$ invariant. Thus by Obs. \ref{obs:approx-sat}, all clusters satisfy the $(2q+1)^{th}$ invariant at the beginning of the $(q+1)^{th}$ application. Assuming that\footnote{Indeed in \cite{chang2020distributed}, it is assumed that $\epsilon<1/K$ for a sufficiently large constant $K$.} $\epsilon\leq 1/(2\alpha)$, after $7$ applications, we get that $t(Q,i)\leq \epsilon_i^5 \cdot \Delta$ for every $i \in [2,\ell]$.
\end{proof}

\subsubsection{Coloring Large Blocks of Layer 1}\label{sec:large-col-one}
Next, we consider the most challenging coloring step concerning the large blocks in layer-1. This subset $V^L_1$ has no guaranteed slack in their palettes, and the ordering in which the algorithm colors these nodes plays a critical role. We present a $O(1)$-time algorithm that colors a subset of the nodes in these large blocks such that the degree of the remaining uncolored subgraph of $V^L_1$ is poly-logarithmic. This subgraph will be then colored within
$O(\log\log\log n)$ rounds.

\begin{theorem}\label{thm:det-coloring-large-one}
There exists a deterministic procedure that colors the $V^L_1$ nodes within $O(\log\log\log n)$ rounds.
\end{theorem}

Due to Lemma \ref{lem:low-deg}, it will be sufficient to show the following.
\begin{lemma}\label{thm:det-coloring-large-one-bound}
There exists a deterministic procedure that colors a subset of nodes in $V^L_1$ such that the maximum degree of the remaining uncolored subgraph induced on $V^L_1$ is poly-logarithmic.
\end{lemma}
We will again be using Framework \ref{frame:dcs}, and must choose the bounds within the invariants appropriately.

Let $S=V^L_1$ where $S = S_1 \cup \ldots \cup S_g$ is a union of clusters.
Since all the blocks belong to layer-$1$, the bounds $D^{(k)}_j,L^{(k)}_j, U^{(k)}_j$ are the same for all the
$S_j$ clusters in this class. We can therefore omit the subscript, and simply write $D^{(k)},L^{(k)}, U^{(k)}$.
At the high level, the proof of this part is very similar to the previous section. We will be using the same bounds for the invariants as in Lemma 4.5 of \cite{chang2020distributed}, up to a constant factor increase, of roughly $\gamma^{O(1)}$, in the upper bounds on the cluster size and on the anti and external degrees. Within $O(1)$ number of iterations (of derandomizing the $\DenseColoringStep$ procedure, all these bounds become poly-logarithmic which establishes the proof of Theorem \ref{thm:det-coloring-large-one}.
For $q \in [1,9]$, define the bounds $D^{(k)}, L^{(k)}$ and $U^{(k)}$ as in Eq. (\ref{eq:bounds-even}) and (\ref{eq:bounds-odd}).
%
For $q=10$, define
$$D^{(2q-1)}=\Theta(\max\{\log^{18}\Delta,\log n\}), ~~L^{(2q-1)}=\Theta(\Delta^{1/10}\log^{17}\Delta) \mbox{~and~} U^{(2q-1)}_j=\Theta(\Delta^{1/10}\log^{18}\Delta)~,$$
and $D^{(2q)}=\Theta(\log n), ~~L^{(2q)}=\Theta(\Delta^{1/20}\log\Delta) \mbox{~and~} U^{(2q)}=\Theta(\Delta^{-1/20}\log^{5c}\Delta)~.$ Thus $\delta^{(2q-1)}=\Theta(\Delta^{-1/10}\log^2 \Delta)$ and
$\delta^{(2q)}=\Theta(\Delta^{-1/20}\log^{-18} \Delta)$.
For $q=11$, define
$$D^{(2q-1)}=\gamma \cdot D^{(2q-2)} \mbox{~and~} U^{(2q-1)}=U^{(2q-2)}~,$$
thus, $\delta^{(2q-1)}=\gamma \cdot \delta^{(2q-2)}$. In addition,
$$D^{(2q)}=\Theta(\log n), ~~L^{(2q)}=\Theta(\log^{5c}n/\log \Delta),~ U^{(2q)}=\Theta(\log^{5c}n) \mbox{~and~} \delta^{(2q)}=\Theta(\log^{-3c}n).$$
Finally, for $q=12$, define
$$D^{(2q-1)}=\gamma \cdot D^{(2q-2)}, ~~L^{(2q-1)}=L^{(2q-2)} \mbox{~and~} U^{(2q-1)}=U^{(2q-2)}~.$$

\begin{definition}[The $k^{th}$ invariant]\label{def:k-invar-one}
A cluster $S_j$ satisfies the $k^{th}$ \textbf{self-invariant} if:

\begin{itemize}
	\item $|S_j| \in [L^{(k)}, U^{(k)}]$, and
	\item each node in $S_j$ has uncolored antidegree at most $D^{(k)}$.
\end{itemize}

In addition, it satisfies the $k^{th}$ \textbf{neighbor-invariant} if

\begin{itemize}
	\item each node in $S_j$ has at most $D^{(k)}$ uncolored neighbors in $N^*(v) \setminus S_j$.
\end{itemize}

A cluster $S_j$ is said to satisfy the $k^{th}$ \textbf{invariant} if it satisfies both the self- and neighbor-invariants.
\end{definition}

Note that, again, this definition fits into the form described in Section \ref{sec:no-slack}; in fact, it is slightly simpler because it does not need the bound $t(k,i)$ on uncolored layer-$i$ neighbors in the neighbor-invariant (so we can consider each $t(k,i)$ to be set to $\infty$).

We start by showing that, again, the $\DenseColoringStep$ procedure (version two) satisfies the properties of Framework \ref{frame:dcs}. 

\begin{lemma}\label{lem:rand-one}
Let $k=2q-1$ for some integer $q \in [1,11]$. Then, applying $\DenseColoringStep$ for coloring $V^L_{2+}$ satisfies the properties of Framework \ref{frame:dcs}. Moreover, the algorithm has only $2$ rounds of communication in the cluster-graph.
\end{lemma}

\def\APPENDRANDONE{
\begin{proof}[Proof of Lemma \ref{lem:rand-one}]
The randomized algorithm is given by the $\DenseColoringStep$ procedure (version two), by applying the algorithm only to the nodes in the clusters of $\mathcal{S}$. For $q \in [1,9]$, the proof follows in an almost identical manner to the proof of Lemma \ref{lem:rand-one}. We therefore need to consider $q \in \{10,11\}$ (for which Lemma 4.5 of \cite{chang2020distributed} provides a special treatment).
Following \cite{chang2020distributed}, for $q=10$, the non-default shrinking rate is $\delta^{(2q-1)}=\Delta^{-1/20}$. However it still holds that $U^{(2q)}=\Theta(\delta^{(2q-1)} \cdot U^{(2q-1)}$ and similarly that
$L^{(2q)}=\Theta(\delta^{(2q-1)} \cdot L^{(2q-1)}$. We have that $U^{(2q)}=\Theta(\Delta^{1/20}\cdot \log^{18}\Delta)$ and $L^{(2q)}=\Theta(\Delta^{1/20}\cdot \log^{17}\Delta)$. To obtain high probability guarantee, the bound on the external and antidegrees is set to $D^{(2q)}=\Theta(\log n)$. In a similar manner, for $q=11$, the shrinking rate is again set to $\delta^{(2q-1)}=\Theta(\Delta^{-1/20}\log^{5c} n)$ and the bounds on the cluster size follow the default equation of $U^{(2q)}=\Theta(\delta^{(2q-1)} \cdot U^{(2q-1)}$ and $L^{(2q)}=\Theta(\delta^{(2q-1)} \cdot L^{(2q-1)}$. Thus, $U^{(2q)}=\log^{5c} n$ and $L^{(2q)}=\log^{5c} n/\log \Delta$. Again, $D^{(2q)}$ is set to be $O(\log n)$. Altogether, since all the bounds are at least logarithmic, by applying Lemma \ref{lem:dense-step-var-two} on sets of logarithmic size, the requirements holds w.h.p. as in the proof of Lemma \ref{lem:rand-two}.
\end{proof}
}

Proof of this lemma is again deferred to Appendix \ref{sec:missing-derand}. We are now ready to complete the proof of Lemma \ref{thm:det-coloring-large-one-bound}.
\begin{proof}[Proof of Theorem \ref{thm:det-coloring-large-one-bound}]
The proof follows simply by having $Q=12$ applications of the algorithm of Theorem \ref{thm:det-one-iteration} which is based on the existence of the randomized algorithm of Lemma \ref{lem:rand-one} for coloring $V^L_{1}$.
First observe that Observation \ref{obs:approx-sat} also holds for under Def. \ref{def:k-invar-one}.
We prove by the induction on $q \in \{1,\ldots, Q\}$ that at the beginning of the $q^{th}$ application all nodes
clusters $\mathcal{S}$ satisfy the $(2q-1)^{th}$ invariant. The base case holds trivially as all clusters in $\mathcal{S}$ are large blocks. Assume that it holds up to $q\geq 1$ and consider the $(q+1)^{th}$ application.
By the induction assumption, at the beginning of the $q^{th}$ application,  all clusters satisfy the $(2q-1)^{th}$ invariant. By applying the algorithm of Theorem \ref{thm:det-one-iteration}, all clusters $\gamma$-satisfy the $(2q)^{th}$ invariant. Thus by Obs. \ref{obs:approx-sat}, all clusters satisfy the $(2q+1)^{th}$ invariant at the beginning of the $(q+1)^{th}$ application.  The theorem then follows as for $Q=11$, we have that $L^{2Q}, U^{2Q},D^{2Q}=\Theta(\log n)$. Note that since all nodes are in layer $1$, $N^*(v)$ correspond to the number of uncolored neighbors in $V^L_1$. Therefore, the maximum degree of the remaining uncolored graph is polylogarithmic.
\end{proof}

\begin{proof}[Proof of Theorem \ref{thm:det-coloring-large-one}]
The coloring of the remaining uncolored nodes in $V^L_1$ is completed by applying the algorithm of Lemma \ref{lem:low-deg}. This algorithm runs in $O(\log\log\log n)$ rounds as desired. The Theorem follows.
\end{proof}


\subsection{Sparse Coloring}\label{sec:sparse}
Let $U$ be the remaining uncolored dense nodes, and recall that $V_{sp}$ is the collection of sparse nodes. The sets $U$ and $V_{sp}$ are colored in \cite{chang2020distributed} by applying an $O(\log^* \Delta)$-round randomized algorithm that w.h.p. colors all nodes in $V_{sp}$ and $U$. 

The main goal of this section is to derandomize this algorithm, within the same order of the number of rounds.
\begin{theorem}[Derandomization of Lemma 2.1 of \cite{chang2020distributed}]\label{thm:derand-sparse}
Consider a directed acyclic graph $G'$, where each node $v$ is associated with a parameter $p_v \leq |\Pal(v)|-\deg(v)$. Let $d^*$ be the maximum out-degree of the graph, let $p^*=\min_{v \in V} p_v$. Suppose that $d^*, p^* \geq \log^c n$, and that there is a number $C=\Omega(1)$ such that every node $v$ satisfies $\sum_{u \in N_{out}(v)} 1/p_u\leq 1/C$. Then, there is a deterministic algorithm for coloring $G'$ within $O(\log^*(p^*)-\log^*C+1)$ rounds. 
\end{theorem}
As in the CLP algorithm, the coloring of the sets $U$ and $V_{sp}$ follow by a direct application of Theorem \ref{thm:derand-sparse}, and therefore we focus on proving Theorem \ref{thm:derand-sparse}. The randomized algorithm of Lemma 2.1 from \cite{chang2020distributed} is based on the following simple $\ColorBidding$ procedure: Each node selects $C/2$ colors from its available colors randomly. node $v$ successfully colors itself if at least one of its selected colors is not in conflict with any other color selected by its neighbors in $N_{out}(v)$. It is then shown that each node gets colored with probability of $exp(-\Omega(C))$, and therefore the gap between the excess colors and the out-degree is increased by an $exp(\Omega(C))$ factor. This exponential improvement in each application of the $\ColorBidding$ procedure leads to the desired round complexity of $O(\log^*(p^*)-\log^*C+1)$ rounds. 

Our derandomization will be based on $Q=O(\log^*(p^*)-\log^*C)$ iterations where at the beginning of each iteration $k \in [1,Q]$, all nodes $v$ will be required to satisfy that $\sum_{u \in N^k_{out}(v)}1/p_u^k \leq 1/C^{(k)}$ where $N^k_{out}(v)$ is the current set of uncolored outgoing neighbors of $v$ at the beginning of the $k^{th}$ iteration, and for every $u$, $p^k_u$ is the current number of $u$'s excess colors (i.e., number of available colors in $u$'s palette minus its uncolored out-degree).  From that point on, we denote the quantity $con(u,k)=\sum_{u \in N^k_{out}(v)}1/p_u^k$ by the \emph{contention} of $v$ in iteration $k$.
We set the sequence of required contention bounds, as in the proof of Lemma 2.1 in \cite{chang2020distributed}.
Initially, $C^{(1)}=\{C, \sqrt{p^*}\}$, i.e., a sufficiently large constant as given by the statement of Theorem \ref{thm:derand-sparse}. For every $k \in [2,Q]$, define: 
$$C^{(k)}=\min\{\sqrt{p^*}, \frac{C^{(k-1)}}{(1+\lambda)exp(-C^{(k-1)}/6)}\}, \mbox{~~and~} Q=\min\{k ~\mid~ C^{(k)}=\sqrt{p^*}\}~,$$
where $\lambda>0$ is a sufficiently small constant. 
\begin{definition}[The $k^{th}$ Invariant]
A node $v$ satisfies the $k^{th}$ invariant, if its contention is at most $con(u,k) \leq 1/C^{(k)}$. 
\end{definition}

In \cite{chang2020distributed}, it was shown that given that all the nodes satisfy the $k^{th}$ invariant, by applying a single application of the $\ColorBidding$ procedure all nodes end up satisfying the $(k+1)^{th}$ invariant, w.h.p.
\begin{lemma}[Slight restatement of Lemma 5.1 of \cite{chang2020distributed}]\label{lem:rand-sparse-step}
Consider a partially colored graph in which every uncolored node has contention at most $1/C$, maximum uncolored out-degree $d^*=\Omega(\log^c n)$ and minimum slack $p^*=\Omega(\log^c n)$. Then after a single application of the $\ColorBidding$ procedure, every node remains uncolored with probability of at most $2exp(-C/6)$, and consequently, the contention of each uncolored node becomes at most 
$2exp(-C/6)/C$ w.h.p.
\end{lemma}

\begin{lemma}\label{lem:derand-sparse-step}[Derandomization of Lemma \ref{lem:rand-sparse-step}]
Consider a partially colored digraph in which every uncolored node satisfies the $k^{th}$ invariant, have maximum uncolored out-degree $d^*=\Omega(\log^c n)$ and a minimum slack $p^*=\Omega(\log^c n)$. Then, there is a $O(1)$-round deterministic \MPC\ algorithm that extends the coloring, such that all the remaining uncolored nodes satisfy the $(k+1)^{th}$ invariant.
\end{lemma}
\begin{proof}
Throughout this proof, we only refer to the directed subgraph induced on the uncolored nodes.
We say that a node $u$ is \emph{happy} if it satisfies the $(k+1)^{th}$ invariant.  By Lemma \ref{lem:rand-sparse-step}, there is a $O(1)$-round randomized local algorithm (given by the $\ColorBidding$ procedure) that makes every uncolored node $u$ happy w.h.p. Also, this algorithm is nice. 

The deterministic algorithm has $Q=\lceil 1/\alpha \rceil$ phases, where each phase satisfies the $(k+1)^{th}$ invariant for roughly $(1-1/n^{\alpha})$-fraction of the yet unhappy nodes. Let $V_q$ be the set of currently unhappy nodes at the beginning of phase $q \in \{1,\ldots, Q\}$, where initially $V_1=V$. Note that the contention of every node $v \in V_q$ is at most $1/C^{(k)}$, since all nodes initially satisfy the $k^{th}$ invariant, and the contention \emph{cannot increase} throughout the coloring procedure. 

We first claim that after running a single application of procedure $\ColorBidding$ on all remaining uncolored nodes,
w.h.p. all nodes satisfy the $(k+1)^{th}$ invariant. This follows immediately by the CLP analysis, and as we as 
the values on $C^{(k)}$ and $C^{(k+1)}$ in the exact same manner. The only useful property that we use here is the fact that the contention is not increasing, and therefore the re-application of the procedure on the yet-unhappy nodes has the same success guarantee as in the initial application (on all the nodes). 
%

By using the PRG construction of Lemma \ref{lem:prg-alg}, a random seed of length $d<\sparam\log n$ can be used to produce $\poly(\Delta)$ pseudorandom coins that $\epsilon$-fool all the randomized $\poly(\Delta)$-time computations for $\epsilon=1/n^{\alpha}$. Therefore, using $d$-length random seed to simulate a single application of the $\ColorBidding$ procedure, leads to at least $(1-1/n^{\alpha})$ fraction of happy nodes. 
Then, by simulating the procedure on each possible seed locally, the machines can compute a $d$-length seed that match this expected value in $O(1)$ rounds. This can be done by allocating a machine $M_u$ for every unhappy and uncolored node $u$ and collecting the $2$-hop ball of $u$ into $M_u$. This allows $M_u$ to evaluate the happiness of $u$ under every possible $d$-length seed. Hence, within $O(1)$ deterministic rounds, we get a subset $V'_{q}\subseteq V_q$ of happy nodes such $|V'_q|\geq (1-1/n^{\alpha})|V_q|$. The input to the next phase $q+1$ is given by $V_{q}\setminus V'_q$. We then have that after $Q$ phases, $V_Q=\emptyset$ and thus all nodes are happy and satisfy the $(k+1)^{th}$ invariant. 
\end{proof}
We are now ready to complete the proof of Theorem \ref{thm:derand-sparse}.
\begin{proof}[Proof of Theorem \ref{thm:derand-sparse}]
The algorithm first repetitively applies the deterministic procedure of Lemma \ref{lem:derand-sparse-step} for $Q=O(\log^*(p^*)-\log^*C+1)$ applications. After these applications, we have that all nodes satisfy the $Q^{th}$ invariant where $C^{(Q)}=\Omega(\sqrt{p^*})$ and for each uncolored node $v$, $con(u,Q)\leq \Omega(1/p^*)$. 

By Lemma \ref{lem:rand-sparse-step}, given that each node $v$ has contention at most $D$, the $\ColorBidding$ procedure colors $v$ with probability of $exp(-\Omega(D))+exp(-\Omega(p^*))$. Therefore for $D=1/\sqrt{p^*}$, every node gets colored by the $\ColorBidding$ procedure, \emph{with high probability}, since $p^*\geq \log^c n$. In our derandomization, we will employ $O(1)$ repetitions of the pseudorandom implementation of the $\ColorBidding$ procedure. Each repetition colors $1-1/n^{\alpha}$ fraction of the yet uncolored nodes. We again exploit the fact that contention of a node is non-decreasing. That is, any legal extension of the current coloring cannot increase the contention. We next describe this final step in more details. 

The final step of the algorithm starts with all nodes having contention bounded by at most $D=1/\sqrt{p^*}$. 
Using the PRG machinery, there is a randomized algorithm that uses a seed of length $d<\sparam\log n$ that colors every remaining node with probability of $1-1/n^{\alpha}$. By allocating a machine for each possible $d$-length seed, the machine can pick a seed that matches this expected value. The algorithm employs $O(1/\alpha)$ iterations of that procedure. Since the contention of a node is non-increasing, each iteration colors $1-1/n^{\alpha}$ fraction of the nodes, and thus after $O(1/\alpha)$ iterations all nodes are colored.  
\end{proof}

\begin{corollary}\label{cor:sparse-gap-col}[Derandomization of Lemma 2.2 of \cite{chang2020distributed}]
Suppose $|\Pal(v)|\geq (1+\rho)\Delta$ for each node $v$, and let $\rho=\Omega(1)$ and $\Delta\geq \log^c n$. Then, there is a deterministic algorithm that takes $O(\log^* \Delta-\log^*\rho+1)$ time to color all nodes in $O(1)$ number of rounds.
\end{corollary}
\begin{proof}
We apply the algorithm of Theorem \ref{thm:derand-sparse} with the following input. Orient the graph edges in an arbitrary manner, and set $p_v=\rho \Delta$ for each $v$. Set $C=\rho, p^*=\rho \cdot \Delta$ and $d^*=\Delta$. The time complexity is $O(\log^*(p^*)-\log^*(C)+1)=O(\log^*\Delta-\log^* \rho+1)$ as desired. 
\end{proof}

\begin{corollary}[Coloring remaining uncolored dense nodes]\label{cor:col-dense-U}
There is a $O(\log^* \Delta)$-round deterministic procedure for coloring the subset $U$ of the remaining uncolored dense nodes.
\end{corollary}
\begin{proof}
We employ the sparse coloring procedure of Theorem \ref{thm:derand-sparse} in a similar manner to \cite{chang2020distributed}. Note that all layer-1 nodes are already colored by Theorem \ref{thm:dense-col-slack-main-one} and Theorem \ref{thm:det-coloring-large-one}. Let $G'$ be a directed acyclic graph induced by $U$ where all edges are oriented from the sparser to the denser endpoint. That is, for a layer-$i$ node $u$ and a layer-$j$ node $v$, an edge $\{u,v\}$ is oriented as $(u,v)$ if $i>j$ or if $i=j$ and $ID(u)>D(v)$. By Theorems \ref{thm:dense-col-slack-main-two} and \ref{thm:col-large-two}, the out-degree of a layer-$i$ node $v$ in $G'$ is at most $\sum_{j=2}^{i}\epsilon_i^5\Delta=O(\epsilon^{5/2}_{i-1}\Delta)$. By Theorem \ref{thm:one-shot-almost}, $p_v=|\Pal(v)|-\deg(v)=\Omega(\epsilon_{i-1}^2\Delta)$. Therefore there is a constant $C$ such that
$$\sum_{u \in N_{out}(v)} 1/p_u  \leq \sum_{j=2}^{i}O\left(\frac{\epsilon_{j-1}^{5/2}\Delta}{\epsilon_{j-1}^2 \Delta}\right)<1/C~.$$
In addition, $p^*=\Omega(\Delta^{8/10})$, $d^*=\Delta$ and $C=\Omega(1)$. The algorithm of Theorem \ref{thm:derand-sparse} colors $G'$ within $O(\log^* \Delta)$ number of rounds as required. 
\end{proof}

\begin{corollary}[Coloring sparse nodes]\label{cor:col-sparse-final}
There is a $O(\log^*\Delta)$-round deterministic procedure for coloring the subset $V_{sp}$ of sparse nodes.
\end{corollary}
\begin{proof}
Again we use the algorithm of Theorem \ref{thm:derand-sparse} with the following input. Let $G''$ be any acyclic orientation of the graph induced by $V_{sp}$. (e.g., orienting an edge towards the endpoint of lower-ID). For every $v \in V_{sp}$, we have that $p_v=|\Pal(v)|-\deg(v)\geq \gamma \cdot \Delta$ for some constant $\gamma=O(\epsilon_\ell^2)$. We therefore have that the contention of $v$ is at most $\sum_{u \in N_{out}(v)}(1/p_u)\leq 1/\gamma$. Thus, the input to the algorithm of Theorem \ref{thm:derand-sparse} is given by $C=1/\gamma$ and $p^*,d^*=\Theta(\Delta)$. The algorithm of Theorem \ref{thm:derand-sparse} then colors $G''$ within $O(\log^* \Delta)$ rounds. 
\end{proof}

We are now ready to complete the proof of Theorem \ref{thm:det-coloring-main}.
\begin{proof}[Proof of Theorem \ref{thm:det-coloring-main}]
First assume that $\Delta$ is polylogarithmic. In this case the coloring is completed in $O(\log\log\log n)$ rounds by applying the algorithm of Lemma \ref{lem:low-deg}. Next assume that $\Delta\geq \log^c n$ for which we provide the derandomization of the CLP algorithm. This algorithm partitions the nodes into dense and sparse nodes. 
The coloring of the dense nodes in layer $1$ is computed within $O(\log\log \log n)$ rounds by Theorem \ref{thm:dense-col-slack-main-one} and Theorem \ref{thm:det-coloring-large-one}. 
The coloring of the remaining dense nodes is completed by Cor. \ref{cor:col-dense-U} within $O(\log^* \Delta)$ rounds. The coloring of the sparse nodes is completed by Cor. \ref{cor:col-sparse-final} within $O(\log^* \Delta)$ rounds. The theorem follows. 
\end{proof}


%
%
%


\newcommand{\Proc}{Proceedings of the~}
\newcommand{\APPROX}{International Workshop on Approximation Algorithms for Combinatorial Optimization Problems (APPROX)}
\newcommand{\CACM}{Communication of the ACM}
\newcommand{\CSR}{International Computer Science Symposium in Russia (CSR)}
\newcommand{\DISC}{International Symposium on Distributed Computing (DISC)}
\newcommand{\FOCS}{IEEE Symposium on Foundations of Computer Science (FOCS)}
\newcommand{\ICALP}{Annual International Colloquium on Automata, Languages and Programming (ICALP)}
\newcommand{\IPCO}{International Integer Programming and Combinatorial Optimization Conference (IPCO)}
\newcommand{\IPL}{Information Processing Letters}
\newcommand{\ISAAC}{International Symposium on Algorithms and Computation (ISAAC)}
\newcommand{\JACM}{Journal of the ACM}
\newcommand{\JCSS}{Journal of Computer and System Sciences}
\newcommand{\NIPS}{Conference on Neural Information Processing Systems (NeurIPS)}
\newcommand{\OPODIS}{International Conference on Principles of Distributed Systems (OPODIS)}
\newcommand{\OSDI}{Conference on Symposium on Opearting Systems Design \& Implementation (OSDI)}
\newcommand{\PODS}{ACM SIGMOD Symposium on Principles of Database Systems (PODS)}
\newcommand{\PODC}{ACM Symposium on Principles of Distributed Computing (PODC)}
\newcommand{\RSA}{Random Structures \& Algorithms}
\newcommand{\SICOMP}{SIAM Journal on Computing}
\newcommand{\SIJDM}{SIAM Journal on Discrete Mathematics}
\newcommand{\SIROCCO}{International Colloquium on Structural Information and Communication Complexity (SIROCCO)}
\newcommand{\SODA}{Annual ACM-SIAM Symposium on Discrete Algorithms (SODA)}
\newcommand{\SPAA}{Annual ACM Symposium on Parallelism in Algorithms and Architectures (SPAA)}
\newcommand{\STACS}{Annual Symposium on Theoretical Aspects of Computer Science (STACS)}
\newcommand{\STOC}{Annual ACM Symposium on Theory of Computing (STOC)}
\newcommand{\TALG}{ACM Transactions on Algorithms}
\newcommand{\TCS}{Theoretical Computer Science}

	
\bibliographystyle{alpha}
\bibliography{references}	

\newcommand{\etalchar}[1]{$^{#1}$}
\begin{thebibliography}{GMR{\etalchar{+}}12}

\bibitem[ANOY14]{ANOY14}
Alexandr Andoni, Aleksandar Nikolov, Krzysztof Onak, and Grigory Yaroslavtsev.
\newblock Parallel algorithms for geometric graph problems.
\newblock In {\em \Proc 46th \STOC}, pages 574--583, 2014.

\bibitem[ASW19]{ASW19}
Sepehr Assadi, Xiaorui Sun, and Omri Weinstein.
\newblock Massively parallel algorithms for finding well-connected components
  in sparse graphs.
\newblock In {\em \Proc 38th \PODC}, pages 461--470, 2019.

\bibitem[BBD{\etalchar{+}}19]{BBDFHKU19}
Soheil Behnezhad, Sebastian Brandt, Mahsa Derakhshan, Manuela Fischer,
  MohammadTaghi Hajiaghayi, Richard~M. Karp, and Jara Uitto.
\newblock Massively parallel computation of matching and {MIS} in sparse
  graphs.
\newblock In {\em \Proc 38th \PODC}, pages 481--490, 2019.
\newblock A preliminary version of a merge of CoRR abs/1807.06701 and CoRR
  abs/1807.05374.

\bibitem[BDE{\etalchar{+}}19]{BDELM19}
Soheil Behnezhad, Laxman Dhulipala, Hossein Esfandiari, Jakub {\L}\c{a}cki, and
  Vahab~S. Mirrokni.
\newblock Near-optimal massively parallel graph connectivity.
\newblock In {\em \Proc 60th \FOCS}, pages 1615--1636, 2019.

\bibitem[BEPS16]{BEPS16}
Leonid Barenboim, Michael Elkin, Seth Pettie, and Johannes Schneider.
\newblock The locality of distributed symmetry breaking.
\newblock {\em Journal of the ACM}, 63(3):20:1--20:45, 2016.

\bibitem[BKM20]{BKM20}
Philipp Bamberger, Fabian Kuhn, and Yannic Maus.
\newblock Efficient deterministic distributed coloring with small bandwidth.
\newblock In {\em \Proc 39th \PODC}, pages 243--252, 2020.

\bibitem[BR94]{BR94}
Mihir Bellare and John Rompel.
\newblock Randomness-efficient oblivious sampling.
\newblock In {\em \Proc 35th \FOCS}, pages 276--287, 1994.

\bibitem[BRRY14]{braverman2014pseudorandom}
Mark Braverman, Anup Rao, Ran Raz, and Amir Yehudayoff.
\newblock Pseudorandom generators for regular branching programs.
\newblock {\em SIAM Journal on Computing}, 43(3):973--986, 2014.

\bibitem[CDP20a]{CDP20a}
Artur Czumaj, Peter Davies, and Merav Parter.
\newblock Graph sparsification for derandomizing massively parallel computation
  with low space.
\newblock In {\em \Proc 32nd \SPAA}, pages 175--185, 2020.

\bibitem[CDP20b]{CDP20b}
Artur Czumaj, Peter Davies, and Merav Parter.
\newblock Simple, deterministic, constant-round coloring in the congested
  clique.
\newblock In {\em \Proc 39th \PODC}, pages 309--318, 2020.

\bibitem[CDP21]{CDP21a}
Artur Czumaj, Peter Davies, and Merav Parter.
\newblock Component stability in low-space massively parallel computations.
\newblock In {\em \Proc 40th \PODC}, pages 481--491, 2021.
\newblock Also in CoRR abs/2106.01880, 2021.

\bibitem[CFG{\etalchar{+}}19]{CFGUZ19}
Yi-Jun Chang, Manuela Fischer, Mohsen Ghaffari, Jara Uitto, and Yufan Zheng.
\newblock The complexity of {$(\Delta+1)$} coloring in {Congested Clique},
  massively parallel computation, and centralized local computation.
\newblock In {\em \Proc 38th \PODC}, pages 471--480, 2019.

\bibitem[CKP19]{CKP19}
{Yi-Jun} Chang, Tsvi Kopelowitz, and Seth Pettie.
\newblock An exponential separation between randomized and deterministic
  complexity in the {LOCAL} model.
\newblock {\em SIAM Journal on Computing}, 48(1):122--143, 2019.

\bibitem[CLP18]{ChangLP18}
Yi{-}Jun Chang, Wenzheng Li, and Seth Pettie.
\newblock An optimal distributed {$(\Delta+1)$}-coloring algorithm?
\newblock In {\em \Proc 50th \STOC}, pages 445--456, 2018.

\bibitem[CLP20]{chang2020distributed}
Yi-Jun Chang, Wenzheng Li, and Seth Pettie.
\newblock Distributed {($\Delta$+1)}-coloring via ultrafast graph shattering.
\newblock {\em SIAM Journal on Computing}, 49(3):497--539, 2020.

\bibitem[CPS17]{CPS17}
Keren {Censor-Hillel}, Merav Parter, and Gregory Schwartzman.
\newblock Derandomizing local distributed algorithms under bandwidth
  restrictions.
\newblock In {\em \Proc 31st \DISC}, pages 11:1--11:16, 2017.

\bibitem[DETT10]{DeETT10}
Anindya De, Omid Etesami, Luca Trevisan, and Madhur Tulsiani.
\newblock Improved pseudorandom generators for depth 2 circuits.
\newblock In {\em \Proc 13th International Conference on Approximation
  Algorithms for Combinatorial Optimization Problems (APPROX) and of the 10th
  the International Conference on Randomization and Computation (RANDOM)},
  pages 504--517, 2010.

\bibitem[GGJ20]{GGJ20}
Mohsen Ghaffari, Christoph Grunau, and Ce~Jin.
\newblock Improved {MPC} algorithms for {MIS}, matching, and coloring on trees
  and beyond.
\newblock In {\em \Proc 34th \DISC}, volume 179, pages 34:1--34:18, 2020.

\bibitem[GGK{\etalchar{+}}18]{GGKMR18}
Mohsen Ghaffari, Themis Gouleakis, Christian Konrad, Slobodan Mitrovi\'c, and
  Ronitt Rubinfeld.
\newblock Improved massively parallel computation algorithms for {MIS},
  matching, and vertex cover.
\newblock In {\em \Proc 37th \PODC}, pages 129--138, 2018.

\bibitem[GJN20]{GhaffariJN20}
Mohsen Ghaffari, Ce~Jin, and Daan Nilis.
\newblock A massively parallel algorithm for minimum weight vertex cover.
\newblock In {\em \Proc 32nd \SPAA}, pages 259--268, 2020.

\bibitem[GK20]{GhaffariKuhnDetCol21}
Mohsen Ghaffari and Fabian Kuhn.
\newblock Deterministic distributed vertex coloring: Simpler, faster, and
  without network decomposition.
\newblock {\em CoRR}, abs/2011.04511, 2020.

\bibitem[GKU19]{GKU19}
Mohsen Ghaffari, Fabian Kuhn, and Jara Uitto.
\newblock Conditional hardness results for massively parallel computation from
  distributed lower bounds.
\newblock In {\em \Proc 60th \FOCS}, pages 1650--1663, 2019.

\bibitem[GMR{\etalchar{+}}12]{GopalanMRTV12}
Parikshit Gopalan, Raghu Meka, Omer Reingold, Luca Trevisan, and Salil~P.
  Vadhan.
\newblock Better pseudorandom generators from milder pseudorandom restrictions.
\newblock In {\em \Proc 53rd \FOCS}, pages 120--129, 2012.

\bibitem[GNT20]{Ghaffari0T20}
Mohsen Ghaffari, Krzysztof Nowicki, and Mikkel Thorup.
\newblock Faster algorithms for edge connectivity via random 2-out
  contractions.
\newblock In {\em \Proc 2020 \SODA}, pages 1260--1279, 2020.

\bibitem[GSZ11]{GSZ11}
Michael~T. Goodrich, Nodari Sitchinava, and Qin Zhang.
\newblock Sorting, searching, and simulation in the {MapReduce} framework.
\newblock In {\em \Proc 22nd \ISAAC}, pages 374--383, 2011.

\bibitem[GU19]{GU19}
Mohsen Ghaffari and Jara Uitto.
\newblock Sparsifying distributed algorithms with ramifications in massively
  parallel computation and centralized local computation.
\newblock In {\em \Proc 30th \SODA}, pages 1636--1653, 2019.

\bibitem[HSS16]{harris2016distributed}
David~G. Harris, Johannes Schneider, and Hsin-Hao Su.
\newblock Distributed {$(\Delta+1)$}-coloring in sublogarithmic rounds.
\newblock In {\em \Proc 48th \STOC}, pages 465--478, 2016.

\bibitem[KPP20]{KothapalliPP20}
Kishore Kothapalli, Shreyas Pai, and Sriram~V. Pemmaraju.
\newblock Sample-and-gather: Fast ruling set algorithms in the low-memory {MPC}
  model.
\newblock In {\em \Proc 40th {IARCS} Annual Conference on Foundations of
  Software Technology and Theoretical Computer Science (FSTTCS)}, pages
  28:1--28:18, 2020.

\bibitem[KSV10]{KSV10}
Howard~J. Karloff, Siddharth Suri, and Sergei Vassilvitskii.
\newblock A model of computation for {MapReduce}.
\newblock In {\em \Proc 21st \SODA}, pages 938--948, 2010.

\bibitem[Kuh09]{Kuhn09}
Fabian Kuhn.
\newblock Weak graph colorings: Distributed algorithms and applications.
\newblock In {\em \Proc 21st \SPAA}, pages 138--144, 2009.

\bibitem[Lin92]{Linial92}
Nathan Linial.
\newblock Locality in distributed graph algorithms.
\newblock {\em SIAM Journal on Computing}, 21(1):193--201, February 1992.

\bibitem[LPPP03]{LotkerPPP03}
Zvi Lotker, Elan Pavlov, Boaz Patt{-}Shamir, and David Peleg.
\newblock {MST} construction in {$O(\log\log n)$} communication rounds.
\newblock In {\em \Proc 15th \SPAA}, pages 94--100, 2003.

\bibitem[LPPP05]{LotkerPPP05}
Zvi Lotker, Boaz Patt{-}Shamir, Elan Pavlov, and David Peleg.
\newblock Minimum-weight spanning tree construction in {$O(\log\log n)$}
  communication rounds.
\newblock {\em SIAM Journal on Computing}, 35(1):120--131, 2005.

\bibitem[Lub86]{Luby86}
Michael Luby.
\newblock A simple parallel algorithm for the maximal independent set problem.
\newblock {\em SIAM Journal on Computing}, 15(4):1036--1053, 1986.

\bibitem[Lub93]{Luby93}
Michael Luby.
\newblock Removing randomness in parallel computation without a processor
  penalty.
\newblock {\em \JCSS}, 47(2):250--286, 1993.

\bibitem[MRT19]{MekaRT19}
Raghu Meka, Omer Reingold, and Avishay Tal.
\newblock Pseudorandom generators for width-3 branching programs.
\newblock In {\em \Proc 51st \STOC}, pages 626--637, 2019.

\bibitem[NW88]{NisanW88}
Noam Nisan and Avi Wigderson.
\newblock Hardness vs. randomness (extended abstract).
\newblock In {\em \Proc 29th \FOCS}, pages 2--11, 1988.

\bibitem[Ona18]{Onak18}
Krzysztof Onak.
\newblock Round compression for parallel graph algorithms in strongly sublinear
  space.
\newblock {\em CoRR abs/1807.08745}, 2018.

\bibitem[Par18]{Parter18}
Merav Parter.
\newblock {$(\Delta+1)$} coloring in the {Congested Clique} model.
\newblock In {\em \Proc 45th \ICALP}, pages 160:1--160:14, 2018.

\bibitem[PS18]{ParterSu18}
Merav Parter and {Hsin-Hao} Su.
\newblock Randomized {$(\Delta+1)$}-coloring in {$O(\log^*\Delta)$} {Congested}
  {Clique} rounds.
\newblock In {\em \Proc 32nd \DISC}, pages 39:1--39:18, 2018.

\bibitem[RG20]{RG20}
V{\'a}clav Rozho{\v{n}} and Mohsen Ghaffari.
\newblock Polylogarithmic-time deterministic network decomposition and
  distributed derandomization.
\newblock In {\em \Proc 52nd \STOC}, pages 350--363, 2020.

\bibitem[Vad12]{Vadhan12}
Salil~P. Vadhan.
\newblock Pseudorandomness.
\newblock {\em Foundations and Trends in Theoretical Computer Science},
  7(1-3):1--336, 2012.

\end{thebibliography}


\newpage
	
\appendix

\section{Pseudocode for CLP Procedures}\label{sec:CLP-code}

For completeness, we provide the pseudocode for coloring procedures from \cite{chang2020distributed} that are used in our algorithm as well. Define $N^*(v)=\{u \in N(v) ~\mid~ ID(u)<ID(v)\}$. 

\begin{mdframed}[hidealllines=false]
$\OneShotColoring$. Each uncolored vertex $v$ decides to participate independently with probability $p$. Each participating vertex $v$ selects a color $c(v)$ from its palette $\Psi(v)$ uniformly at random. A participating vertex $v$ successfully colors itself if $c(v)$ is not chosen by any vertex in $N^*(v)$.
\end{mdframed}

\begin{mdframed}[hidealllines=false]
$\DenseColoringStep$ (version 1). 
\begin{enumerate}
\item Let $\pi:\{1,\ldots, |S_j|\}\to S_j$ be the unique permutation that lists the vertices of $S_j$ in increasing order by layer number, breaking ties (within the same layer) by IDs. For $q$ from 1 to $|S_j|$, the vertex $\pi(q)$ selects a color $c(\pi(q))$ uniformly at random from
$$\Pal(\pi(q)) \setminus \{c(\pi(q')) ~\mid~ q' < q \mbox{~and~} \{\pi(q), \pi(q')\}\in E(G)\}~.$$

\item Each $v \in S_j$ permanently colors itself $c(v)$ if $c(v)$ is not selected by any vertices in $N_{out}(v)$.
\end{enumerate}
\end{mdframed}

\begin{mdframed}[hidealllines=false]
$\DenseColoringStep$ (version 2). 
\begin{enumerate}
\item Each cluster $S_j$ selects $(1-\delta_j)|S_j|$ vertices uniformly at random and generates a permutation $\pi$ of those vertices uniformly at random. The vertex $\pi(q)$ selects a color $c(\pi(q))$ uniformly at random from
$$\Pal(\pi(q)) \setminus \{c(\pi(q')) ~\mid~ q' < q \mbox{~and~} \{\pi(q), \pi(q')\}\in E(G)\}~.$$

\item Each $v \in S_j$ permanently colors itself $c(v)$ if $c(v)$ is not selected by any vertices in $N_{out}(v)$.
\end{enumerate}
\end{mdframed}

In the $\ColorBidding$ procedure each vertex $v$ is associated with a parameter $p_v \geq |\Psi(v)|-N_{out}(v)$ and $p^*=\min_v p_v$. Let $C$ be a constant satisfying that $\sum_{u \in N_{out}(v)}1/p_u\leq 1/C$. All vertices agree on the value $C$.
\begin{mdframed}[hidealllines=false]
$\ColorBidding$. Each color $c \in \Psi(v)$ is added to $S_v$ with probability $C/2p_v$ independently. If there exists a color $c^* \in S_v$ that is not selected by any of the vertices in $N_{out}(v)$, $v$ colors itself $c^*$.
\end{mdframed}

\APPENDLOWDEGCOMP

\section{Generating Initial Slack: Proof of Theorem \ref{thm:one-shot-almost}}\label{sec:genslack}

To prove Theorem \ref{thm:one-shot-almost}, we will partition nodes into groups in such a way that we ensure that the local sparsity property is maintained \emph{within each group}. We will then apply the PRG to directly derandomize Lemma \ref{thm:one-shot-random} within each group.  

Recall that, since we performed Linial's coloring algorithm, we can assume nodes have IDs in $[poly(\Delta)]$, such that no two nodes within $t$ hops share the same ID, for some constant $t$. Our first step is to partition nodes to $4C$ groups, for sufficiently large constant $C$, using a $4z$-wise independent hash function $h:[poly(\Delta)]\rightarrow [4C]$ for some sufficiently large $z=\Theta(\log n/\log \Delta)$. By Lemma \ref{lem:hash}, families $\mathcal H$ of such functions exist of size $n^{O(1)}$, i.e. any function from the family can be specified with an $O(\log n)$-bit seed.

Fix a particular node $v$ and a group $i \in [4C]$ to consider, and let $\epsilon$ be the highest value such that $v$ is $\epsilon$-sparse. We make the following definition to distinguish between two separate cases in the analysis:

\begin{definition}
We call $v$ type-1 if $\deg(v)\le (1-\frac{\eps}{2})\Delta$, and type-2 otherwise.
\end{definition}

Type-1 nodes already have sufficient slack, and we aim to ensure that that their uncolored degree remains sufficiently high, while for type-2 nodes, we also aim to preserve sufficient anti-edges between $v$'s neighbors (i.e., pairs of $v$'s neighbors without an edge between them) in the group, in order to create slack. We define some notation for quantities we wish to control:

\begin{definition}
For fixed node $v$, group $i$, and hash function $h$, let

\begin{itemize}
	\item $d_i$ denote the number of $v$'s neighbors hashed to group $i$,
	\item Let $\Psi$ denote the set of non-$\eps$-friend neighbors of $v$ (before hashing), and $\psi$ its size,
	\item $\Psi_i$ denote the set of neighbors of $v$ hashed to group $i$, which have at least $\frac{\eps\Delta}{96C^2}$ anti-edges to other neighbor of $v$ in group $i$, and $\psi_i$ its size. 
\end{itemize}
\end{definition}
We show a lemma controlling the number of $v$'s neighbors hashed to group $i$. This will provide the property we need for type-1 nodes, but will also be useful for type-2 nodes.

\begin{lemma}
For any node $v$ and group $i$, $|d_i - \frac{\deg(v)}{4C}|\le \deg(v)^{0.5}\Delta^{0.01}$ with probability at least $1-n^{-4}$, under a random hash function $h\in \mathcal H$.
\end{lemma}

\begin{proof}
We apply Lemma \ref{lem:conc} with $Z_1, \dots, Z_{\deg(v)}$ as indicator random variables for the events that each neighbor of $v$ is hashed to group $i$. These variables are $4z$-wise independent and each have expectation $\frac{1}{4C}$. Since we assume $z=\Theta(\log n/\log \Delta)$ with sufficiently large constant factor, by Lemma \ref{lem:conc}, and $\Delta \ge \log^{c} n$ for sufficiently $c$, we obtain 
\begin{align*}
\Prob{|d_i - \frac{\deg(v)}{4C}| \ge \deg(v)^{0.5}\Delta^{0.01}} 
&\le 8 \left(\frac{\frac{4z\deg(v)}{4C} + 16 z^2}{\deg(v)\Delta^{0.02}}\right)^{2z}\\
&\le 8 \left(zC^{-1}\Delta^{-0.02} + 16z^2\Delta^{-0.02}\right)^{2z} \\
&\le n^{-c}\enspace.
\end{align*}
\end{proof}

We now show the property we need about type-2 nodes. Our aim is to bound the number of non-friend-neighbors of $v$, hashed to group $i$, that are still non-friends (under a different parameter $\eps'$) when restricting to group $i$ (i.e. they have sufficient anti-edges to other neighbors of $v$ also in group $i$).

To do so, we study the behavior of a proxy set $S$ of anti-edges in $v$'s neighborhood. Note that each non-$\eps$-friend of $v$ has at least $\deg(v)-(1-\eps)\Delta$ anti-edges to other neighbors of $v$ by definition; for type-2 node $v$ this is at least $\frac{\eps}{2}\Delta$. We construct $S$ as follows: each non-friend neighbor of $v$ `nominates' exactly $\frac{\eps}{2}\Delta$ adjacent anti-edges (to other neighbors of $v$) to add to $S$ (choosing arbitrarily if it has more than $\eps\Delta$ such anti-edges). The size of $S$ is then between $\frac{\eps\psi}{4}\Delta$ and $\frac{\eps\psi}{2}\Delta$, since anti-edges can be nominated by both of their endpoints. We then let $S_i$ denote the subset of $S$ containing anti-edges which have both endpoints hashed to group $i$.

We consider the size of $S_i$. For each $s\in S$, let $Z_s$ be the indicator variable for the event that both endpoints of $s$ hash to group $i$. Then, the sum $Z$ of these variables is $|S_i|$, and its expectation $\mu$ is $\frac{S}{(4C)^2}$. We do not, however, have bounded independence between the variables $Z_s$, since if anti-edges $s$ and $s'$ share an endpoint, $Z_s$ and $Z_{s'}$ are dependent. So, we must show the following concentration bound manually:

\begin{lemma}
	$\Prob{|Z-\mu| \ge \Delta^{1.6} }\le n^{-4}$.
\end{lemma}

\begin{proof}
	\begin{align*}
		\Exp{|Z-\mu|^{2z}} &=	\Exp{\left(\sum_{s\in S}Z_s-(4C)^{-2}\right) ^{2z}}
		= \sum_{(s_1,\dots,s_c)\in S^c} \Exp{\prod_{i=1}^{c}( Z_{s_i}-(4C)^{-2})}\enspace.
	\end{align*}
	
	Consider a particular tuple $(s_1,\dots,s_{2z})\in S^c$ of anti-edges. If there is some anti-edge $s_j$ that does not share any endpoints with any other $s_{j'}$, then by independence (the $2z$ anti-edges in the tuple have at most total $4z$ endpoints, and these \emph{are} hashed independently),
	\begin{align*}
		\Exp{\prod_{i=1}^{2z}( Z_{s_i}-(4C)^{-2})} &= 
		\Exp{\prod_{\substack{i=1,\\i\ne j}}^{2z}(Z_{s_i}-(4C)^{-2})}\cdot \Exp{(Z_{s_i}-(4C)^{-2})}\\
		& = \Exp{\prod_{\substack{i=1\\i\ne j}}^{2z}(Z_{s_i}-(4C)^{-2})}\cdot 0 = 0\enspace.
	\end{align*}
	
	So, we need only consider tuples where all anti-edges share an endpoint with at least one other. The set of endpoints of anti-edges in such tuples is of size at most $3z$ (maximized when the anti-edges are in pairs sharing an endpoint, and are otherwise disjoint). There are at most $\Delta^{3z}$ such sets, and each admits at most $\binom{3z}{2}^{2z} \le (3z)^{4z}$ tuples.
	We have  $	\Exp{\prod_{i=1}^{2z}(Z_{s_i}-(4C)^{-2})}\le 1$, so each such tuple contributes at most $1$ to the sum. So, $\Exp{|Z-\mu|^{2z}}\le \Delta^{3z} (3z)^{4z} $.

Then, by Markov's inequality, 
\[\Prob{|Z-\mu| \ge \Delta^{1.6} } \le \Prob{(Z -\mu)^{2z}\ge  \Delta^{3.2z}} \le \frac{\Exp{(Z-\mu)^{2z}}}{\Delta^{3.2z}}\le \Delta^{-0.2z}(3z)^{4z}\le n^{-4}.\]
\end{proof}

In this case, $S_i \ge \frac{|S_v|}{(4C)^2}-\Delta^{1.6} \ge \frac{\psi\eps\Delta}{64C^2} -\Delta^{1.6}$. We use this to obtain a bound on $\psi_i$:

Each anti-edge in $S_i$ is between a non-$\eps$-friend of $v$ and another neighbor of $v$. Since there are $\psi$ non-$\eps$-friends of $v$, at most $\frac{\psi\eps\Delta}{96C^2}$ anti-edges can have been nominated by those not in $\Psi_i$. The remaining $\frac{\psi\eps\Delta}{192C^2}$ anti-edges in $S_i$ must therefore be contributed by non-$\eps$-friends of $v$ in $\Psi_i$. Since each such non-$\eps$-friend of $v$ nominates $\frac{\eps}{2}\Delta$ anti-edges, we have $\psi_i \ge \frac{\psi}{96C^2}$.

\begin{lemma}\label{lem:groupprops1}
In $O(1)$ rounds, we can partition nodes into $4C$ groups such that for each node $v$ and each group $i$,
\[|d_i - \frac{\deg(v)}{4C}|\le \deg(v)^{0.5}\Delta^{0.01}.\]

Furthermore, if $v$ is type-2, 
\[\psi_i \ge \frac{\psi}{96C^2}.\]
\end{lemma}

\begin{proof}
Both of these conditions hold with high probability in $n$ under a uniformly random hash function from $\mathcal H$, and can be checked for a node $v$ by storing only its neighborhood, which fits on a machine. The hash functions can be specified using $O(\log n)$ bits. Therefore, we can apply the method of conditional expectations to fix a hash function satisfying the properties for all nodes, in $O(1)$ low-space \MPC rounds.
\end{proof}

We now have the following properties:

\begin{lemma}\label{lem:groupprops2}
After partitioning into $4C$ groups, 
	\begin{itemize}
		\item the maximum degree in each group is at most $\Delta' := \frac{\Delta}{4C} + \Delta^{0.51}$,
		\item any node $v$ with $\deg(v)\ge 0.9\Delta$ has $d_i \ge (5/6)\Delta'$ for each group $i$, and
		\item any type-2 node $v$ that was $\eps$-sparse is $\eps' := \frac{\eps}{400C^2}$-sparse in each group.
	\end{itemize}
\end{lemma}

\begin{proof}
The first point follows since, by Lemma \ref{lem:groupprops1}, for each node $v$ and group $i$ we have $d_i\le \frac{\deg(v)}{4C}+\deg(v)^{0.5}\Delta^{0.01} \le \frac{\Delta}{4C}+\Delta^{0.51} $. The second follows since, if $\deg(v)\ge 0.9\Delta$, then 
\[d_i\ge \frac{\deg(v)}{4C}-\deg(v)^{0.5}\Delta^{0.01}
\ge \frac{0.9\Delta}{4C}-\Delta^{0.51} \ge \frac{(5/6)\Delta}{4C} + (5/6)\Delta^{0.51} = (5/6)\Delta'\enspace.
 \]
 
For the third point, by Lemma \ref{lem:groupprops1} $v$ has at least $\frac{\psi}{96C^2}\ge \frac{\eps\Delta}{96C^2}$ neighbors in group $i$ who have at least $\frac{\eps\Delta}{96C^2}$ anti-edges to other neighbors of $v$ in group $i$, i.e. $v$ is $\frac{\eps}{96C^2}$-sparse upon restricting to group $i$.
\end{proof}

Now, we simply apply the PRG to derandomize the coloring procedure of Lemma \ref{thm:one-shot-random} within each of the first $C$ groups in turn, updating node palettes after each. We will deem a node $v$ happy if it indeed receives $\Omega(\eps^2\Delta)$ excess colors. For nodes which originally had $\deg(v)\ge 0.9\Delta$ and so now have $d_i \ge (5/6)\Delta'$ for each $i$, this occurs with probability $1-1/n^c$ \emph{in each group} by Lemma \ref{thm:one-shot-random} (since $\eps' = \Theta(\eps)$). 
\footnote{The instances we are applying to here are not actually $\Delta'$-list coloring instances, but instances in which each node $v$ has palette size at least $\max\{\deg(v)+1, \Delta'-(\Delta')^{0.6}\}$; as shown in \cite{Parter18} (see Lemma \ref{lem:relaxed-CLP}), the analysis of CLP still holds in this case).} (Nodes with $\deg(v)<0.9\Delta$ already had $\Omega(\Delta) = \Omega(\eps^2\Delta)$ excess colors.)
This is a property that can be verified for a node $v$ by seeing only $v$'s $1$-hop neighborhood and performing a $poly(\Delta)$-time local computation. By Claim \ref{cl:aux}, our PRG derandomization will make all but a $n^{-\alpha}$-fraction of nodes happy, for some sufficiently small constant $\alpha$. In subsequent groups after the first, we consider only the nodes who are not already happy, and so reduce the amount of unhappy nodes by an $n^{\alpha}$ factor per group. After $1/\alpha$ groups we no longer have any unhappy nodes; we therefore set $C=1/\alpha$.

Since all nodes now have the required excess colors, we have satisfied the second condition of Theorem \ref{thm:one-shot-almost}. To see that we also satisfy the first, note that we leave the last $3C$ groups uncolored. For each node $v$, the number of $v$'s neighbors in these $3C$ groups is at least

\[3C (\frac{\deg(v)}{4C}- \deg(v)^{0.5}\Delta^{0.01}) 
\ge \frac{3}{4}\deg(v)- 3C\Delta^{0.51}\enspace,\]
by Lemma \ref{lem:groupprops1}.

So, when $\deg(v)\ge (5/6)\Delta$, $v$ has at least 
$\frac{3}{4}\deg(v)- 3C\Delta^{0.51} \ge 0.625\Delta-3C\Delta^{0.51} > \Delta/2$ uncolored neighbors. We have therefore proven Theorem \ref{thm:one-shot-almost}.

\section{Missing Proofs for the CLP Derandomization}\label{sec:missing-derand}



\APPENDRANDTWO

\APPENDRANDONE

\subsection{Handling Smaller and Unequal Palettes}\label{sec:small-pal}
The following lemma from \cite{Parter18} claims that having (possibly unequal palettes) each of $|\Pal(v)|\geq \max\{\deg_G(v)+1, \Delta-\Delta^{3/5}\}$ has only a minor impact on the CLP algorithm. 
\begin{lemma}\label{lem:relaxed-CLP}\cite{Parter18}
Let $G$ be an $n$-node graph with maximum degree $\Delta$, and suppose that each node $v$ has a palette $\Pal(v)$  that satisfies $|\Pal(v)|\geq \max\{\deg_G(v)+1, \Delta-\Delta^{3/5}\}$. Then the CLP algorithm can be slightly modified to compute list coloring for $G$ within $O(\log^*\Delta)$ pre-shattering \emph{randomized} steps and $Det_d(\poly log n)$ post-shattering \emph{deterministic} steps, where $Det_d(n')$ is the deterministic \local\ complexity of solving the $(\deg+1)$ list coloring problem on $n'$-node graphs. 
\end{lemma}
We briefly outline the proof of \cite{Parter18}. Essentially, the randomized procedures (and thus also their derandomization) remain the same but the bounds on desired antidegrees etc. are affected by a constant factor. 

The key modification is in definition of $\epsilon$-friend (which affects the entire decomposition of the graph). Throughout, let $q=\Delta^{3/5}$ and say that $u,v$ are $(\epsilon,q)$-friends if $|N(u)\cap N(v)|\geq (1-\epsilon)\cdot (\Delta-q)$. Clearly, if $u,v$ are $\epsilon$-friends, they are also $(\epsilon,q)$-friends. A node $v$ is an $(\epsilon,q)$-\emph{dense} if it has at least $(1-\epsilon)\cdot (\Delta-q)$ neighbors which are $(\epsilon,q)$-friends. An $(\epsilon,q)$-almost clique is a connected component of the subgraph induced by $(\epsilon,q)$-dense nodes and their $(\epsilon,q)$ friends edges. 

\begin{observation}\label{obs:nonep}\cite{Parter18}
For any $\epsilon \in [\Delta^{-10},K^{-1}]$, where $K$ is a large constant, and for $q=\Delta^{3/5}$, it holds that if $u,v$ are $(\epsilon,q)$ friends, they are $(2\epsilon)$-friends. Also, if $v$ is an $(\epsilon,q)$-dense, then it is $2\epsilon$-dense. 
\end{observation}
The observation explain the constant factor effect on the CLP bounds. 

\paragraph{Hierarchy and Blocks.}
The entire hierarchy of levels $V_1,\ldots, V_{\ell}$ is based on using the definition of $(\epsilon,q)$-dense nodes (rater than $\epsilon$-dense nodes).
Let $\bar{d}_{S,V'}(v)=|(N(v) \cap V')\setminus S|$ be the external degree of $v$ with respect to $S,V'$. Let $a_S(v)=|S \setminus (N(v)\cup \{v\})|$ be the antidegree of $v$ with respect to $S$. Let $V_{\epsilon,q}^d, V_{\epsilon,q}^s$ be the nodes which are $(\epsilon,q)$-dense (resp., sparse).
By Obs. \ref{obs:nonep}, for every $\epsilon \in [\Delta^{-10},1/K]$, it also holds:
\begin{observation}
For any $(\epsilon,q)$-almost clique $C$, there exists an $\epsilon$-almost clique $C_{\epsilon}$ and an $(2\epsilon)$-almost clique $C_{2\epsilon}$ such that $C_{\epsilon}\subseteq C \subseteq C_{2\epsilon}$. 
\end{observation}
As a result, Lemma \ref{lem:blocks} holds up to small changes in the constants.
The notion of blocks is trivially extended using the definition of $(\epsilon,q)$-dense nodes. The properties of the large, medium and small blocks remain the same (up to constant factor in the proofs), as $(\epsilon,q)$-almost clique is contained in an $2\epsilon$-almost clique. 
Note that since each node has a palette of at least $\deg+1$ colors, the excess of colors is non-decreasing throughout the coloring algorithm.

The coloring probability of Alg. $\OneShotColoring$ holds to be $p \in (0,1/8)$ rather than $p \in (0,1/4)$. 
Coloring the large blocks and the sparse nodes follows the exact same analysis an in \cite{chang2020distributed}. 
The reason is that the coloring of the large blocks uses only the bounds on the external degrees and clique size and those properties hold up to insignificant changes in the constants (the algorithms of the sparse and dense components use the $O$-notation on these bounds). 

\end{document}